\documentclass[a4paper,12pt]{amsart} 

\usepackage[lmargin=2.4cm,rmargin=2.25cm,bmargin=3cm,tmargin=2.5cm]{geometry}

\usepackage[utf8]{inputenc}
\usepackage{amsmath}
\usepackage{amssymb}  
\usepackage{amsthm}
\usepackage{textcomp}
\usepackage{enumerate}
\usepackage[numbers,sort]{natbib}
\usepackage{array}
\usepackage{bbm}
\usepackage{verbatim}
\usepackage{hyperref}
\usepackage{graphicx}
\usepackage{booktabs}
\usepackage{algorithm}
\usepackage{mathrsfs}
\usepackage{amsmath}
\usepackage{amsthm}
\usepackage{amsfonts,mathtools,bm}
\usepackage{amssymb}

\usepackage{xcolor}
\hypersetup{
  colorlinks,
  linkcolor={red!50!black},
  citecolor={blue!50!black},
  urlcolor={blue!80!black}
}

\usepackage{tikz}
\newcounter{hypA}
\newenvironment{hypA}{\refstepcounter{hypA}\begin{itemize}
  \item[({\bf A\arabic{hypA}})]}{\end{itemize}}
\newcounter{hypD}

\newcommand{\defeq}{\mathrel{\mathop:}=} 
\newcommand{\eqdef}{=\mathrel{\mathop:}}

\newcommand{\uarg}{\,\cdot\,}
\newcommand{\ud}{\mathrm{d}}
\newcommand{\R}{\mathbb{R}}
\newcommand{\N}{\mathbb{N}}
\renewcommand{\P}{\mathbb{P}}
\newcommand{\E}{\mathbb{E}}

\newcommand{\var}{\mathrm{var}}

\newcommand{\bigmid}{\;\big|\;}
\providecommand{\abs}[1]{\lvert#1\rvert}

\def\bbR{\mathbb{R}}

\def\bbR{\mathbb{R}}
\def\bbP{\mathbb{P}}
\def\bbE{\mathbb{E}}

\theoremstyle{plain}
\newtheorem{theorem}{Theorem}
\newtheorem{corollary}[theorem]{Corollary}
\newtheorem{proposition}[theorem]{Proposition}
\newtheorem{lemma}[theorem]{Lemma}

\theoremstyle{definition}

\newtheorem{assumption}[theorem]{Assumption}

\newtheorem{recommendation}{Recommendation}

\newtheorem{assumptionx}{Condition} 

\theoremstyle{remark}
\newtheorem{remark}[theorem]{Remark}

\newcommand{\rate}{\ensuremath{\rho}}
\newcommand{\costrate}{\ensuremath{\gamma}}

\newcommand{\nrt}{n} 
\newcommand{\nrp}{N} 
\newcommand{\nri}{{m_{\mathrm{iter}}}} 
\newcommand{\nrijump}{{m_{\mathrm{iter}}^{\mathrm{jump}}} } 
\newcommand{\nrii}{m} 
\newcommand{\nriii}{m} 

\newcommand{\xss}{x} 
\newcommand{\xps}{\pmb{x}} 
\newcommand{\xssc}{\check{x}} 

\newcommand{\funlatentmain}{\varphi} 
\newcommand{\funlatenthelp}{g} 
\newcommand{\sterm}{\tilde{\Delta}} 

\newcommand{\funfullmain}{f} 
\newcommand{\funfullhelp}{g} 

\newcommand{\sss}{E} 
\newcommand{\ssss}{\mathcal{E}} 

\newcommand{\xhmm}{\mathbf{X}} 
\newcommand{\xhmmc}{\mathbf{\check{X}}} 

\newcommand{\gammass}{\gamma} 
\newcommand{\gammasstrue}{\gamma^{(\infty)}} 
\newcommand{\gammaps}{\pmb{\gamma}} 
\newcommand{\etaps}{\pmb{\eta}}

\newcommand{\mss}{M} 
\newcommand{\mssc}{\check{M}} 
\newcommand{\msstrue}{M^{(\infty)}}

\newcommand{\gss}{G} 
\newcommand{\gssc}{\check{G}} 

\newcommand{\stheta}{\mathsf{T}} 
\newcommand{\thetaone}{\theta}
\newcommand{\thetatwo}{\theta}

\newcommand{\mcmclatent}{R} 
\newcommand{\cost}{\mathscr{C}} 
\newcommand{\lmcmc}{\mathscr{L}} 
\newcommand{\taufun}{\pmb{\tau}}

\newcommand{\pimcmc}{\Pi}
\newcommand{\mIRE}{\mathbb{E}[\pmb{\tau}]} 

\newcommand{\increment}{\delta W}

\title[Unbiased inference for diffusions]{Unbiased inference for discretely observed hidden Markov model diffusions}

\author[N. K. Chada, J. Franks, A. Jasra, K. J. H. Law \& M. Vihola]{Neil K. Chada, Jordan Franks, Ajay Jasra, Kody J. H. Law \&  Matti Vihola}
\address{N. K. Chada \& A. Jasra, Applied Mathematics and Computational Science,
King Abdullah University of Science and Technology, Thuwal, 23955, KSA. Email: \texttt{neil.chada@kaust.edu.sa} \& \texttt{ajay.jasra@kaust.edu.sa}  }
\address{J. Franks, School of Mathematics, Statistics and Physics,
  Herschel Building, NE1 7RU Newcastle University, UK.  Email: \texttt{jordan.franks@newcastle.ac.uk}}
\address{M. Vihola, Department of Mathematics and Statistics
P.O.Box 35, FI-40014 University of Jyv\"askyl\"a, FI. Email: \texttt{matti.s.vihola@jyu.fi}}
\address{K. J. H. Law, School of Mathematics,
University of Manchester, Manchester, M13 9PL, UK. Email: \texttt{kodylaw@gmail.com}}
\begin{document}
\begin{abstract} 
  We develop a Bayesian inference method for diffusions observed discretely and with noise, which is free of discretisation bias. Unlike existing unbiased inference methods, our method does not rely on exact simulation techniques.  Instead, our method uses standard time-discretised approximations of diffusions, such as the Euler-Maruyama scheme.  Our approach is based on particle marginal Metropolis-Hastings, a particle filter, randomised multilevel Monte Carlo, and importance sampling type correction of approximate Markov chain Monte Carlo.  The resulting estimator leads to inference without a bias from the time-discretisation as the number of Markov chain iterations increases.  We give convergence results and recommend allocations for algorithm inputs.  Our method admits a straightforward parallelisation, and can be computationally efficient.  The user-friendly approach is illustrated on three examples, where the underlying diffusion is an Ornstein-Uhlenbeck process, a geometric Brownian motion, and a $2d$ non-reversible Langevin equation.
  \end{abstract} 

\maketitle
\section{Introduction}\label{sec:intro} 

Hidden Markov models (HMMs) are widely used in real applications, for example, for financial and physical systems modeling; see \cite{cappe}. We focus on the case where the hidden Markov chain arises from a diffusion process that is observed with noise at some number of discrete points in time; see e.g.~\cite{sorensen}.
The parameters associated to the model are static and assigned a prior density. 
Bayesian inference involves expectations with respect to (w.r.t.) the joint posterior distribution of parameters and states, and is important in problems of model calibration and uncertainty quantification. A difficult aspect of Bayesian inference for these models is simulation or evaluation of the diffusion dynamics. Unless the transition probability is explicitly known (see Section 4.4 of \cite{kloeden-platen}),  one often resorts to time discretisation, leading to biased inference.  
If one is to seek inference where such discretisation is avoided, there are broadly two schemes (an exception is \cite{fearnheadLRS}, which we shall discuss below) which one can
follow, the first are the elegant exact simulation methods in \cite{beskos-papaspiliopoulos-roberts-fearnhead,beskos-roberts,fearnheadPR,wang-rao-teh} or the debiasing methods of  
\cite{mcleish,rhee-glynn}, which is the direction followed in this article. We call inference when there is no time-discretisation error, \emph{unbiased inference} and this is the objective of this article.

Approaches to unbiased inference using the afore-mentioned exact simulation methodology apply for a certain class of diffusions. First, the existence of the Lamperti transformation, after which the process has unit diffusion matrix, and second, the drift of the transformed process has to be of gradient form. Although this includes some important diffusion processes, often, these conditions do not hold for multivariate diffusion processes; this limits the scope of the application of these novel schemes. A more recent and general approach can be found in \cite{fearnheadLRS},
which focuses on a type of continuous-time importance sampling method for continuous-time Markov processes, including diffusion processes. The method is, in essence, a type of continuous-time sequential importance sampling algorithm that produces a signed approximation of laws of diffusion processes. Whilst the methodology applies to a reasonably wide class of diffusions, the signed approximation can introduce several algorithmic issues. This includes a large cost associated to the final time of the diffusion process, which at the very least can lead to linear-in-time errors. This latter issue is problematic in our context, as if the methodology is used for HMMs, the time will relate to the number of data and the errors reported in \cite{fearnheadLRS} will be prohibitive for our application.
It should also be noted that this methodology requires a good proposal based upon a process which is analytically soluble. This limits the applicability of the method, which will not be the case in our context. As the utility of the method has not been explored for the problem of interest in this article, we have considered an alternative method. We proceed with an Euler-Maruyama time-discretisation (see \cite{kloeden-platen}), referred henceforth as \emph{Euler} which is generally applicable and is combined with the debiasing schemes of \cite{mcleish,rhee-glynn}.

Traditional inference approaches based on time-discretisations face a trade-off between bias and computational cost. 
Once the user has decided on a suitably fine discretisation size, one can run, for example, the particle marginal Metropolis-Hastings (PMMH) \cite{andrieu-doucet-holenstein}.  This algorithm uses a particle filter (PF) (see \cite{del-moral}), where proposals between time points are generated by the approximation scheme, and ultimately accepted or rejected according to a Metropolis-Hastings type acceptance ratio; see \cite{golightly-wilkinson-bayesian-2011}.  As the discretisation size adopted must be quite fine, a PMMH algorithm can be computationally intensive.

To deal with the computational cost of PMMH, \cite{jasra-kamatani-law-zhou} develop a PMMH based method
which uses (deterministic) multilevel Monte Carlo (dMLMC) \cite{giles-or,heinrich}. The basic premise of MLMC is to introduce a telescoping sum representation of the posterior expectation associated to the most precise time
discretisation. Then, given an appropriate coupling of posteriors with `consecutive' time discretisations, the cost associated to a target mean square error is reduced,
relative to exact sampling from the most precise (time-discretised) posterior. In the HMM diffusion context, the standard MLMC method (for diffusions without observations) is not applicable, so based upon a PF coupling approach and PMMH, an MLMC method is devised in \cite{jasra-kamatani-law-zhou,jasra2018multi}, which achieves fine-level, though biased, inference. 

\subsection{Method}
The unbiased and computationally efficient inference method suggested in this paper is built firstly on PMMH, using Euler type discretisations, but using a PMMH targeting a coarse-level model, which is less computationally expensive. This does not yield unbiased inference yet, but it can be achieved by an importance sampling (IS) type correction; see \cite{vihola-helske-franks}.

We suggest an IS type correction that is based on a single-term (randomised) MLMC type estimator \cite{mcleish,rhee-glynn} and the PF coupling approach of \cite{jasra-kamatani-law-zhou}.
The rMLMC correction is based on randomising the running level in the multilevel context of a certain PF, which we refer to as the `delta PF ($\Delta$PF)' (Algorithm \ref{alg:debiased-pf}).  In short, the $\Delta$PF uses the PF coupling introduced in \cite{jasra-kamatani-law-zhou}, but here an estimator is used for unbiased estimation of the difference of \emph{unnormalised} integrals corresponding to two consecutive discretisation levels, over the latent states with parameter held fixed (see Section \ref{sec:delta-pf}). 
In \cite{jasra-kamatani-law-zhou},
each term in the difference of PMMH averages 
is individually self-normalised (at each level) because of the unknown normalising constants.

The resulting IS type estimator leads to unbiased inference over the joint posterior distribution, and is highly parallelisable, as the more costly (randomised) $\Delta$PF corrections may be performed independently \emph{en masse} given the PMMH base chain output.  We are also able to suggest optimal choices for algorithm inputs in a straightforward manner (Recommendation \ref{rec:allocations} and Figure \ref{fig:optimal}).  This is because there is no bias, and therefore the difficult cost-variance-bias trade-off triangle associated with dMLMC is not present.  Besides being unbiased and efficient, our method is user-friendly, as it is a combination of well-known and relatively straightforward components: PMMH, Euler approximations, PF, rMLMC, and an IS type estimator.
For more about the strengths of the method, see Remark \ref{alg:is-mlmc-remark} later, as well as \cite{vihola-helske-franks,franks-vihola} for more discussion about IS (type) estimators based on approximate Markov chain Monte Carlo (MCMC).

Key to verifying consistency of the method is a finite variance assumption for the r$\Delta$PF estimator.  We verify a parameter-uniform bound for the variance under a simple set of HMM diffusion conditions in Section \ref{sec:bound}.
Note, however, that consistency of our method is likely to hold more generally.  This is in contradistinction to methods based on exact simulation, which require analytically tractable transformations to unit covariance diffusion term and computable bounds in the rejection sampler, in order to even apply the method (see for example the review in the recent preprint \cite{wang-rao-teh}).

We consider a non-reversible  Langevin equation  in Section \ref{sec:experiments}, where, to the authors' best knowledge, exact simulation is not applicable.
If an exact simulation method is applicable, the obvious question arises whether our method or the exact simulation method should be applied.  The efficiency of exact simulation type methods is dependent upon several and different factors than our method.  These factors for exact simulation include proper tuning and tight computable bounds for the rejection sampler.  In an ideal scenario for exact simulation, a method based on exact simulation is likely to perform better than our method.  However, in the reverse case, our method can perform better, if the efficiency of exact simulation is poor.  For instance, the efficiency of exact simulation decreases to zero as the analytically computed upper bound of the IS weight used in the rejection sampler  increases to infinity. 

We remark that in principle our algorithm may require simulations with
arbitrarily fine discretisation sizes and corresponding arbitrarily large cost;
this can be infeasible. 
However, it should be noted that the user specifies the chance that this might occur, 
so one can ensure that the probability of 'very expensive' simulations is 
arbitrarily small.
In addition, our method typically has finite expected cost, or cost that is finite with high probability: we give conditions (see Section \ref{sec:bound}), typical elsewhere in the context of HMM diffusions \cite{del-moral,kloeden-platen}, which ensure this (see Section \ref{sec:efficiency}).

Although we have mostly in mind the case of Euler approximation schemes for the diffusion dynamics approximation, which are generally implementable, other schemes could be possibly be used as well; see \cite{giles-szpruch}.  However, suitable couplings for these schemes in dimensions greater than one may not be trivial.  For the sake of theory and proof of consistency, ideally these would have also known weak and strong order convergence rates; see \cite{kloeden-platen}.
Indeed, assuming a coupling exists, such higher-order schemes can improve convergence of our method (see Sections \ref{sec:efficiency} and \ref{sec:experiments}).  More generally, our approach based on PMMH or other approximate MCMC, increasingly fine families of approximations, MLMC, and IS correction, could be applied beyond the HMM diffusion context, for example, to HMM jump-diffusions; see \cite{jasraLO}.    
\subsection{Outline}
Section \ref{sec:delta-pf} introduces the aforementioned $\Delta$PF (Algorithm \ref{alg:delta-pf}) and subsequently discusses some applications of randomisation techniques.  The
theoretical properties of the $\Delta$PF in the HMM diffusion context are summarised in Section \ref{sec:bound}. Section \ref{sec:method} presents the suggested IS type estimator (Algorithm \ref{alg:is-mlmc}), based on PMMH with rMLMC (i.e.~r$\Delta$PF) correction, and details its consistency and a corresponding central limit theorem (CLT). Section \ref{sec:efficiency} suggests suitable allocations in the $\Delta$PF based on rMLMC efficiency considerations. The numerical experiments in Section \ref{sec:experiments} illustrate our method in practice in the setting of an Ornstein-Uhlenbeck process, geometric Brownian motion, and a non-reversible Langevin equation.
Proofs for the technical results of Sections \ref{sec:bound}, \ref{sec:method} and \ref{sec:efficiency} are given in Appendix \ref{app:bound}, \ref{app:method} and \ref{app:efficiency}, respectively.


\subsection{Notation}

Let $(\sss, \ssss)$ be a measurable space.  Functions $\varphi:\sss\to\R$ will be assumed measurable.
We denote by $\mathscr{P}(\sss)$ the collection of probability measures on $(\sss,\ssss)$, and by $\mathcal{B}_b(\sss)$ the set of $\varphi:\sss\rightarrow \mathbb{R}$ with 
$\|\varphi\|\defeq\sup_{x\in E}|\varphi(x)|<\infty$.
For a measure $\mu$ on $(\sss,\ssss)$, we set $\mu(\varphi)\defeq\int_{\sss}\varphi(x)\mu(\ud x)$ whenever well-defined.
For $K:\sss\times\sss\rightarrow[0,1]$ a Markov kernel and $\mu\in \mathscr{P}(\sss)$, we set
$
\mu K(\ud y) \defeq \int_{\sss}\mu(\ud x) K(x,\ud y),
$
and 
$
K(\varphi)(x) \defeq \int_{\sss} \varphi(y) K(x,\ud y),
$
whenever well-defined.
We use the convention $\prod_{\emptyset} \defeq 1$, and $p{:}q\defeq \{r\in \mathbb{Z}: p\le r \le q\}$.

\section{Delta particle filter for unbiased estimation of level differences}
\label{sec:delta-pf} 

Consider the (It\^{o}) diffusion process
\begin{eqnarray}\label{eq:sde}
\ud Z_t & = & a_{\thetaone}(Z_t)\ud t + b_{\thetaone}(Z_t)\ud W_t, \qquad t\ge 0,
\end{eqnarray}
with $Z_t\in\mathsf{X}\defeq \mathbb{R}^d$, model parameter $\theta\in \stheta$, $a:\mathsf{X}\times\mathsf{T}\rightarrow\mathsf{X}$, $b:\mathsf{X}\times\mathsf{T}\rightarrow\mathsf{X}\times\mathsf{X}$, $\{W_t\}_{t\ge 0}$ a $d-$dimensional Brownian motion, and the initial value $Z_0=z_0\in\mathsf{X}$ a fixed value.
We suppose that there are data 
$\{Y_p=y_p\}_{p=0}^{\nrt}$, $y_p \in \bbR^m$,
which are observed at equally spaced discrete times, $1{:}(\nrt+1)$ for simplicity. 
We shall consider the discrete time skeleton of the diffusion \eqref{eq:sde} $Z_1,\dots,Z_{n+1}$,
where we shall set $X_p=Z_{p+1}$, for $p=0{:}\nrt$.
The Markov transition between $X_{p-1}$ and $X_p$, $p=1{:}n$, 
is given by the transition kernel $M^{(\theta,\infty)}(x_{p-1}, \ud x_p)$ of the diffusion process over unit time, 
with initial
distribution $\eta_0^{(\theta,\infty)}(dx_0)=M^{(\theta,\infty)}(z_0,dx_0)$.
It is assumed that conditional on $X_{p}$, 
$Y_p$ is independent of random variables $\{X_i, Y_i\}_{i\neq p}$ and has density $g_{\thetatwo}(y_p|x_{p})\eqdef G_p^{(\thetatwo)}(x_p)$.  
Setting $M_0^{(\theta,\infty)}=\eta_0^{(\theta,\infty)}$ and $M_p^{(\theta,\infty)}=M^{(\theta,\infty)}$, $p=1{:}n$,
the resulting pair $(M_p^{(\theta,\infty)}, G_p^{(\theta)})$ defines the HMM diffusion, and is an example of a so-called \emph{Feynman-Kac model} (see \cite{del-moral}) described below.  As the results of this section can just as easily be stated in terms of Feynman-Kac models, we do so in the following, which shows the generality of our approach.

\begin{remark}\label{rem-eul-hmm}
  In many situations of practical interest, $M^{(\theta,\infty)}(x_{p-1}, \ud x_p)$, $p=1{:}n$ exists, but one cannot even simulate from it and/or evaluate a non-negative unbiased estimator of it.
One can often consider the HMM diffusion where one works with a time discretisation of $M^{(\theta,\infty)}(x_{p-1}, \ud x_p)$. For instance for $\ell\in\{0\}\times\mathbb{N}$, $h_{\ell}=2^{-\ell}$, one
could consider the Euler approximation for $t=p,p+h_{\ell},\dots,p+1-h_{\ell}$, with $Z_{p}=x_{p-1}$ given
$$
 Z_{t+h_{\ell}}^{\ell} = Z_t^{\ell} + a_\theta(Z_t^{\ell}) h_{\ell} + b_\theta(Z_t^{\ell}) \increment_{t+h_{\ell}}^{\ell},
$$
where $\increment_{t+h_{\ell}}^{\ell}\stackrel{i.i.d.}{\sim}\mathcal{N}(0,h_{\ell})$ (Gaussian distribution 0 mean, covariance $h_{\ell} I$) and one will set $X_p=Z_{p+1}^{\ell}$.
The induced transition kernel over unit time is written $M^{(\theta,\ell)}(x_{p-1}, dx_p)$. A similar remark can be made for $\eta_0^{(\theta,\infty)}$ with discretisation $\eta_0^{(\theta,\ell)}$.
\end{remark}

\subsection{Particle filters}
A Feynman-Kac model $(M_n,G_n)$ on spaces $(\sss_n,\ssss_n)$ arises when
\begin{enumerate}[(i)]
\item $\mss_n(\xss_{0:n-1}, \ud \xss_n)$ are (regular) probability `transition' kernels from
  $\sss_{0:\nrt-1}$ to $\sss_\nrt$ for $n\ge 1$, and $M_0(\ud x_0) \defeq \eta_0(\ud \xss_0)\in \mathscr{P}(\sss_0)$, and 
    \item$G_n(\xss_{0:\nrt})$ are $[0,\infty)$-valued (measurable) `potential' functions for $n\ge 0$.
\end{enumerate}
The Particle filter (Algorithm \ref{alg:pf}) (see \cite{del-moral}) generates sets of samples
and weights corresponding to the Feynman-Kac model, which for $\funlatentmain:E_{0:\nrt}\to\R$ lead to an unbiased estimator for the (unnormalised) \emph{smoother} $\gammaps_\nrt(G_\nrt \funlatentmain)$, defined here in terms of the (unnormalised) \emph{predictor} 
\begin{equation}\label{eq:flow}
    \gammaps_\nrt(\funlatentmain) \defeq 
    \int     \funlatentmain(x_{0:\nrt})
    \Big(\prod_{t=0}^{\nrt-1}
    \gss_t( \xss_{0:t})
    \Big)
    \eta_0(\ud \xss_0)
    \prod_{t=1}^\nrt
    \mss_t(\xss_{0:t-1}, \ud \xss_t).
\end{equation}
We remark that Step \ref{it:res_step} of Algorithm \ref{alg:pf} refers to the resampling step, which can be multinomial, residual,
stratified or systematic; see e.g.~\cite{cappe,douc-cappe-moulines}.

\begin{algorithm}
    \caption{Particle filter for a Feynman-Kac model.}
\label{alg:pf} 
      \raggedright
      \textbf{Input:} $(M_{0:\nrt},G_{0:\nrt})\defeq (M_t,G_t)_{t=0:\nrt}$ and $\nrp$ the number of particles.
      \begin{enumerate}[(i)]
      \item For $i=1{:}N$ sample $\xss_0^{(i)} \sim \eta_0(\uarg)$ and set $\xps_0^{(i)}\defeq \xss_0^{(i)}$.
      \item
        For $i=1{:}N$ compute $\omega_{0}^{(i)} \defeq
        G_0(\mathbf{\xps}_{0}^{(i)})$
        and set $\bar{\omega}_0^{(i)} \defeq \omega_0^{(i)}/\omega_0^*$
        where $\omega_0^* = \sum_{j=1}^\nrp
        \omega_{0}^{(j)}$.
      \end{enumerate}
      \raggedright
      For $t=1{:}\nrt$, do: 
      \begin{enumerate}[(i)]
        \stepcounter{enumi}
        \stepcounter{enumi}
      \item \label{item:resampling} 
        Given $\bar{\omega}_{t-1}^{(1:\nrp)}$, sample $A_{t-1}^{(1:\nrp)}\in\{1,\dots,N\}^N$ satisfying
        $\E\big[ \sum_{j=1}^\nrp \mathbf{1}\{A_{t-1}^{(j)}= k\}\big]= \nrp \bar{\omega}_{t-1}^{(k)}$ for all $k\in 1{:}N$.  
        \label{it:res_step}
      \item
        For $i=1{:}N$ sample $\xss_t^{(i)} \sim M_t( \xps_{t-1}^{A_{t-1}^{(i)}}, \uarg)$ and set $\xps_t^{(i)}=(\xps_{t-1}^{(A_{t-1}^{(i)})}, \xss_t^{(i)})$.
        \item 
        For $i=1{:}N$ compute $\omega_t^{(i)} \defeq G_t(\xps_t^{(i)})$
          and set $\bar{\omega}_t^{(i)} \defeq \omega_t^{(i)}/\omega_t^*$
          where $\omega_t^* \defeq \sum_{j=1}^\nrp
          \omega_{t}^{(j)}$.
      \end{enumerate}
      \raggedright
      Set for $i=1{:}N$, $\xhmm^{(i)}\defeq \xps_\nrt^{(i)}$. If $\omega_t^*>0$, for $i=1{:}N$ set 
      $V^{(i)} \defeq \bar{\omega}_\nrt^{(i)} \prod_{t=0}^\nrt \frac{1}{\nrp}
      \omega_t^*$, otherwise, for $i=1{:}N$ set $V^{(i)}=0$.\\
      \textbf{Output:} $(V^{(1:\nrp)},\xhmm^{(1:\nrp)})$. 
\end{algorithm} 
\begin{proposition} 
    \label{prop:pf-unbiased} 
Suppose that 
$\funlatentmain:E_{0:\nrt}\to\R$ is such that $\gammaps_\nrt(G_n \funlatentmain)<\infty$.  Then, the output of Algorithm \ref{alg:pf}
satisfies
\[
    \E\bigg[ \sum_{i=1}^\nrp V^{(i)} \funlatentmain(\xhmm^{(i)}) \bigg]
    = \gammaps_\nrt(\gss_\nrt \funlatentmain).
\]
\end{proposition} 
Proposition \ref{prop:pf-unbiased} is a restatement of Theorem
7.4.2 of \cite{del-moral} in case $A_{t-1}^{(i)}$ are sampled independently
(`multinomial resampling').
The extension to the general unbiased case
is straightforward; see
\cite{vihola-helske-franks}.

\subsection{Level difference estimation} 
Suppose that we have two Feynman-Kac models
$(M_n^F,$ $G_n^F)$ and $(M_n^C, G_n^C)$ defined on common spaces $(\sss_n,\ssss_n)$. The models correspond to `finer' and `coarser'
Euler type discretised HMM diffusions.
We are interested in estimating (unbiasedly) the difference
\begin{equation}\label{eq:smoothing-difference}
\gammaps^F_\nrt(\gss_\nrt^F\funlatentmain)
-\gammaps^C_\nrt(\gss_\nrt^C\funlatentmain). 
\end{equation}
If the models are close to each other, as they will be in the
multilevel (diffusion) context, we would like the estimator
also to be typically small.  In many contexts, if one can estimate the difference using a coupling, it is possible to obtain a variance reduction.  The particular coupling approach we use here is based on using a combined Feynman-Kac model as in \cite{jasra-kamatani-law-zhou}, which provides a simple, general and effective coupling of PFs, and which we will use to estimate the level difference of unnormalised smoother \ref{eq:smoothing-difference}. 

Hereafter, we denote $\xssc_{\nrt} = (\xssc^F_{\nrt},
\xssc^C_{\nrt})\in E_\nrt\times E_\nrt$, and for $\xssc_{0:n}=(\xssc_0,\ldots, \xssc_n)\in \sss_{0}^2\times \ldots \sss_{n}^2$, we set $\xssc_{0:n}^s\defeq ( \xssc_0^s,\ldots, \xssc_n^s)\in \sss_{0:n}$ for $s\in \{F,C\}$.  
\begin{assumption} 
    \label{a:coupled-fk} 
Suppose that $(\mssc_t, \gssc_t)$ is a Feynman-Kac
model on the product spaces $(E_t\times E_t, \ssss_t\otimes \ssss_t)$, such that:
\begin{quote}
\begin{enumerate}[(i)]
    \item 
      \label{a:coupled-fk-kernel} $\mssc_t$ is a 
      coupling of the probability measures $\mss_t^F$ and $\mss_t^C$, i.e.~for all $A\in \ssss_t$, we have 
      $$
      \int_{A\times \sss_t } 
\mssc_t(\xssc_{0:t-1}, \ud \xssc_t) 
= \mss_t^F( \xssc_{0:t-1}^F, A)
,\qquad
    \int_{\sss_t\times A} 
\mssc_t(\xssc_{0:t-1}, \ud \xssc_t) 
= \mss_t^C( \xssc_{0:t-1}^C, A),
$$
and for $A\in \ssss_0$, we have
$
\check{\eta}_0(A\times \sss_0) = \eta_0^F(A)
$
and
$
\check{\eta}_0(\sss_0 \times A) = \eta_0^C(A).
$
\item \label{a:coupled-fk-potential}
  $\gssc_t(\xssc_{0:t}) \defeq \frac{1}{2} \big[\gss_t^F(\xssc^F_{0:t})
+ G_t^C(\xssc^C_{0:t})\big]$. 
\end{enumerate}
\end{quote}
\end{assumption} 

Algorithm \ref{alg:delta-pf} presents a methodology to unbiasedly estimate the level differences \eqref{eq:smoothing-difference}. In the context of hidden Markov model diffusions, we explain in Remark \ref{delta-pf-remark} how to satisfy Assumption \ref{a:coupled-fk}.

\begin{algorithm} 
    \caption{Delta particle filter ($\Delta$PF) for unbiased estimation of level differences.}
    \label{alg:delta-pf} 
\raggedright
    \textbf{Input}: $(\mssc_{0:\nrt}, \gssc_{0:\nrt})$ and $\nrp$ the number of particles.
\begin{enumerate}[(i)]
\item Run Algorithm \ref{alg:pf} with $(\mssc_{0:\nrt}, \gssc_{0:\nrt},\nrp)$, {outputting $(\check{V}^{(1:\nrp)},\xhmmc^{(1:\nrp)})$.}
    \item Compute $(V^{(1:2\nrp)},\xhmm^{(1:2\nrp)})$ where
      \[
          \begin{pmatrix}
              V^{(i)},
              & \xhmm^{(i)} \end{pmatrix}
              \defeq 
          \begin{cases} 
              \begin{pmatrix} \check{V}^{(i)}
                  w^F(\xhmmc^{(i)}),
              &  \xhmmc^{(i)F}
              \end{pmatrix}
              &
          i=1{:}\nrp,\\
          \begin{pmatrix}
          -\check{V}^{(i-\nrp)} w^C(\xhmmc^{(i-\nrp)}),
          & \xhmmc^{(i-\nrp)C}
          \end{pmatrix}
             & i=(\nrp+1){:}2\nrp, 
          \end{cases}
      \]
      and 
      $
          w^F(\xssc_{0:\nrt}) \defeq \frac{\prod_{t=0}^\nrt
            \gss_t^F(\xssc_{0:t}^F)}{\prod_{t=0}^\nrt \gssc_t(\xssc_{0:t})}$
       and 
       $
          w^C(\xssc_{0:\nrt}) \defeq \frac{\prod_{t=0}^\nrt
            \gss_t^C(\xssc_{0:\nrt}^C)}{\prod_{t=0}^\nrt \gssc_t(\xssc_{0:t})}.
      $
\end{enumerate}
\textbf{Output}: $(V^{(1:2\nrp)},\xhmm^{(1:2\nrp)})$.
\end{algorithm}    

\begin{proposition} 
    \label{prop:delta-pf} 
Under Assumption \ref{a:coupled-fk}, the output of Algorithm 
\ref{alg:delta-pf} satisfies
\[
    \E \bigg[ \sum_{i=1}^{2\nrp} V^{(i)} \funlatentmain(\xhmm^{(i)})\bigg]
    =
    \gammaps_\nrt^F(G_\nrt^F \funlatentmain) - \gammaps_\nrt^C(G_\nrt^C \funlatentmain),
\]
whenever both integrals on the right are well-defined and finite.
\end{proposition} 
\begin{proof} 
Applying Assumption \ref{a:coupled-fk} one can use the unbiasedness property of PF in Algorithm \ref{alg:pf}, to yield 
\begin{align*}
    \E\bigg[ \sum_{i=1}^{\nrp} V^{(i)} \funlatentmain(\xhmm^{(i)})\bigg]
    &= \int
    w^F(\xssc_{0:\nrt}) \funlatentmain(\xssc_{0:\nrt}^F)
    \Big(\prod_{t=0}^\nrt
    \gssc_t(\xssc_{0:t})\Big)
    \check{\eta}(\ud x_0)
    \prod_{t=1}^\nrt
 \mssc_t(\xssc_{0:t-1}, \ud \xssc_{t})
     \\
     &= \int  \funlatentmain(\xssc_{0:\nrt}^F)
     \Big(\prod_{t=0}^\nrt
     \gss_t^F(\xssc_{0:t}^F)\Big)
     \eta^F_0(\ud x_0)
     \mss_t^F(\xssc_{0:t-1}^F, \ud \xssc_{t}^F)
    =
    \gammaps_\nrt^F(G_\nrt^F\funlatentmain),
\end{align*}
where, specifically, Assumption \ref{a:coupled-fk}\ref{a:coupled-fk-potential}
guarantees $\gssc_t>0$
whenever $\gss_t^F>0$, and Assumption \ref{a:coupled-fk-kernel} implies the
marginal law of $\prod_{t=0}^\nrt \mssc_t$ is $\prod_{t=0}^\nrt \mss_t^F$.  Similarly,
$
\E\big[ \sum_{i=\nrp+1}^{2\nrp} V^{(i)} \funlatentmain(\xhmm^{(i)})\big] =-\gammaps_n^C(\gss_\nrt^C\funlatentmain).
$
\end{proof} 

\begin{remark}\label{delta-pf-remark}
Regarding Algorithm \ref{alg:delta-pf}:
\begin{quote}
\begin{enumerate}[(i)]
    \item \label{delta-pf-remark-kernel}
   In the hidden Markov model diffusion context, one could consider $F$ to correspond to an Euler discretisation
at level $\ell$ step-size $h_{\ell}=2^{-\ell}$ and $C$  to correspond to an Euler discretisation
at level $\ell-1$ step-size $h_{\ell-1}=2^{-(\ell-1)}$. The couplings 
$\mssc_p$, $p=1{:}n$, (of the two Euler transitions over unit time - see Remark \ref{rem-eul-hmm}) could be based on using the same underlying Brownian motion; see \cite{kloeden-platen}.  That is, for $t=p,p+2h^{\ell},\dots,p+1-2h^{\ell}$,
$(Z_{p}^{\ell},Z_{p}^{\ell-1})=(x_{p-1}^{\ell},x_{p-1}^{\ell-1})$ given,
\begin{align*}
  Z_{t+h_{\ell}}^{\ell} &= Z_t^{\ell} + a_\theta(Z_t^{\ell}) h_{\ell} + b_\theta(Z_t^{\ell}) \increment_{t+h_{\ell}}^{\ell}\\
  Z_{t+2h_{\ell}}^{\ell} &= Z_{t+h_{\ell}}^{\ell} + a_\theta(Z_{t+h_{\ell}}^{\ell}) h_{\ell} + b_\theta(Z_{t+h_{\ell}}^{\ell}) \increment_{t+2h_{\ell}}^{\ell},
\end{align*}
with $\increment_{t+k h_{\ell}}^{\ell}\stackrel{i.i.d.}{\sim} \mathcal{N}(0, h_{\ell})$, $k=1:2$, then we can use
$$
Z_{t+h_{\ell-1}}^{\ell-1} = Z_t^{\ell-1} + a_\theta(Z_t^{\ell-1}) h_{\ell-1} + b_\theta(Z_t^{\ell-1}) \big(\increment_{t+h_{\ell}}^{\ell} + \increment_{t+2h_{\ell}}^{\ell}\big),
$$
for the coarser Euler discretisation.  We set $(X_{p}^{\ell},X_{p}^{\ell-1})=(z_{p+1}^{\ell},z_{p+1}^{\ell-1})$. A similar remark can be made for the initialisation.
The potentials $\gss_t^{\ell}$ and $\gss_t^{\ell-1}$ are then simply the conditional likelihood functions $\gss_t^{\ell}=\gss_t^{\ell-1}=G_t^{(\theta)}$.
\item 
  The choice of $\gssc_t$ in Assumption \ref{a:coupled-fk}\ref{a:coupled-fk-potential} 
  provides a safe `balance' in between the
approximations, as $w^F$ and $w^C$ are upper bounded by $2^{\nrt +1}$.
Indeed, the coupled Feynman-Kac model can be thought
as an `average' of the two extreme cases--with the choice
$\gssc_t(x_{0:t})=\gss_t^F(\xssc^F_{0:t})$ the coupled PF
would coincide marginally with the Feynman-Kac model with dynamics
$M_t^F$. What is the optimal choice for $\gssc_t$ is an interesting question.
\item \label{delta-pf-remark-max}
  Clearly, the choice of $\gssc_{0:t}$ can be made also in other
ways.  It is sufficient
for unbiasedness to choose $\gssc_t(\xssc_{0:t})$ such that it 
is strictly positive whenever either the $\gss_t^F(\xssc_{0:t}^F)$ or $\gss_t^C(\xssc_{0:t}^C)$ product is positive, but choices which make $w^F$ and $w^C$ bounded are safer, for instance $\gssc_{0:t}(\xssc_{0:t}) =
\max\{ \gss_t^F(\xssc_{0:t}^F), \gss_t^C(\xssc_{0:t}^C)\}$.  This was the original choice made in \cite{jasra-kamatani-law-zhou} for approximation of normalised smoother differences. This PF coupling approach based on change of measure and weight corrections $w^F$ and $w^C$, has been further used also, for example, in \cite{jasra2018multi}.
\item 
  Later, in the HMM diffusion context, we set $\gss_t^F = \gss_t^C$, corresponding to common observational densities, but the method is also of interest with differing potentials.
\end{enumerate}
\end{quote}
\end{remark}

\subsection{Unbiased latent inference}
We show here how the randomisation techniques of \cite{mcleish,rhee-glynn} can be used with the output of Algorithm  \ref{alg:pf} and \ref{alg:delta-pf} to provide an unbiased estimator according to the true model, even though the PFs are only run according to approximate models.
Let us index the transitions $\mss^{(\ell)}_p$ and potentials $\gss^{(\ell)}_p$ by $\ell\ge 0$.  They are assumed throughout to be increasingly refined approximations, in the (weak) sense that
\begin{equation}\label{eq:refined}
\gammaps_\nrt^{(\ell)}(G_\nrt^{(\ell)} \funlatentmain)
\longrightarrow
\gammaps_\nrt^{(\infty)}(G_\nrt^{(\infty)} \funlatentmain),
\qquad\text{as}\quad \ell\to\infty,
\end{equation}
for all $\funlatentmain\in \mathcal{B}_b(E_{0:\nrt})$, where
\begin{equation*} 
    \gammaps_\nrt^{(\ell)}(\funlatentmain) \defeq 
    \int     \funlatentmain(x_{0:\nrt})
    \Big(\prod_{t=0}^{\nrt-1}
    \gss_t^{(\ell)}( \xss_{0:t})
    \Big)
    \eta_0^{(\ell)}(\ud \xss_0)
    \prod_{t=1}^\nrt
    \mss_t^{(\ell)}(\xss_{0:t-1}, \ud \xss_t).
\end{equation*}
In Assumption \ref{a:coupled-fk} we set symbols $(F,C)$ to be $(\ell,\ell-1)$ for $\ell\ge 1$. We will write the potentials and kernels of the coupled Feynman-Kac model (in the sense
of Assumption \ref{a:coupled-fk}) as $(\check{M}_{0:n}^{(\ell)},\check{G}_{0:n}^{(\ell)})$. As a result of \ref{lem:debiased-pf}, Algorithm \ref{alg:debiased-pf} can provide unbiased estimation of
$\gammaps_\nrt^{(\infty)}(\gss_\nrt^{(\infty)}\funlatentmain)$, leading to unbiased inference w.r.t.~the normalised smoother
$$
\funlatentmain\mapsto \frac{\gammaps_\nrt^{(\infty)}(\gss_\nrt \funlatentmain) }{\gammaps_\nrt^{(\infty)}(\gss_\nrt)} \eqdef \hat{\etaps}_n^{(\infty)}(\funlatentmain),
$$
which is stated as Proposition \ref{prop:independent-pf} below.

We remark that in step (\ref{alg:debiased-pf-delta}) of Algorithm \ref{alg:debiased-pf}, in principle, one may have to run Algorithm \ref{alg:delta-pf}
for arbitrarily large $L$. However, it should be noted that the user specifies the probability $\mathbf{p}=(p_\ell)_{\ell\in\N}$, so one can ensure the probability of simulating `very large' values of $L$ is arbitrarily small.
\begin{algorithm}
\caption{Unbiased estimator based on PF and r$\Delta$PF. }
\label{alg:debiased-pf}
\raggedright
\textbf{Input:} $\Big((M^{(\ell)}_{0:\nrt}, G^{(\ell)}_{0:\nrt})\Big)_{\ell\in\{0\}\cup\N}$, $\nrp$ the number of particles and probability $\mathbf{p}=(p_\ell)_{\ell\in\N}$.
\begin{enumerate}[(i)]
  \item Run Algorithm \ref{alg:pf} with $(M^{(0)}_{0:\nrt}, G^{(0)}_{0:\nrt},\nrp)$, outputting
    $(V^{(1:\nrp)'}, \xhmm^{(1:\nrp)'})$.
  \item Sample $L \sim \mathbf{p}$, independently from the other random variables.
    \item \label{alg:debiased-pf-delta} Run Algorithm \ref{alg:delta-pf} with $(\check{M}^{(L)}_{0:\nrt}, \check{G}^{(L)}_{0:\nrt},\nrp)$, outputting
  $(V^{(1:2\nrp)}, \xhmm^{(1:2\nrp)})$.
  \end{enumerate}
  \raggedright
  \textbf{Output:} $\big( (V^{(1:\nrp)'}, \xhmm^{(1:\nrp)'}),\, (V^{(1:2\nrp)}, \xhmm^{(1:2\nrp)}), L \big)$.
\end{algorithm}
\begin{assumption}\label{a:abstract-pf} 
  Assumption \ref{a:coupled-fk} holds, $\mathbf{p}=(p_\ell)_{\ell\in \N}$ is a probability on $\mathbb{N}\defeq\mathbb{Z}_{\ge 1}$ with $p_\ell>0$ for all $\ell\ge 1$,
  $\funlatenthelp:\sss_{0:\nrt}\to \R$ is a function, and 
  \begin{equation}\label{eq:variance}
s_\funlatenthelp \defeq \sum_{\ell \ge 0} \frac{ \E \Delta_\ell^2(\funlatenthelp)}{p_\ell} <\infty,
\end{equation}
where
\begin{equation}\label{eq:Delta}
  \Delta_\ell(\funlatenthelp) \defeq \sum_{i=1}^{2\nrp} V^{(i)} \funlatenthelp(\xhmm^{(i)}),
  \end{equation}
  is formed from the output $(V^{(1:2\nrp)}, \xhmm^{(1:2\nrp)})$ of Algorithm \ref{alg:delta-pf} with $(\mssc^{(\ell)}_{0:\nrt}, \gssc^{(\ell)}_{0:\nrt},\nrp)$.
  \end{assumption}

\begin{lemma}\label{lem:debiased-pf}
  Under Assumption \ref{a:abstract-pf}, the estimator
\begin{equation}\label{eq:debiased-pf}
\zeta(\funlatenthelp)\defeq
\sum_{i=1}^\nrp V^{(i)'} \funlatenthelp(\xhmm^{(i)'})
+
\frac{1}{p_{L}} \Delta_L(g),
 \end{equation}
formed from the 
  output of Algorithm \ref{alg:debiased-pf} satisfies
$$
\E[\zeta(\funlatenthelp)]
=   \gammaps_\nrt^{(\infty)}(G_\nrt^{(\infty)} \funlatenthelp),
$$
    whenever $\gammaps_\nrt^{(0)}(G_\nrt \funlatenthelp)$ and $\gammaps_\nrt^{(\infty)}(G_\nrt \funlatenthelp)$ are both finite.
  \end{lemma}
\begin{proof}
 Under Assumption \ref{a:abstract-pf}, 
  we have (see \cite{rhee-glynn,vihola-unbiased})
  $$
  \E[ p_L^{-1} \Delta_L(g)]= \gammaps_\nrt^{(\infty)}(\gss_\nrt^{(\infty)} \funlatenthelp) - \gammaps_\nrt^{(0)}(\gss_n^{(0)} \funlatenthelp),
  $$
    so the result follows by Proposition \ref{prop:pf-unbiased} and linearity of the expectation.
  \end{proof}
The following suggests a fully parallelisable algorithm for unbiased inference over the normalised smoother, and is an unbiased alternative to the particle independent Metropolis-Hastings (PIMH) \cite{andrieu-doucet-holenstein} run at some fine level of discretisation.
\begin{proposition}\label{prop:independent-pf}
  Suppose $\mathbf{p}$ on $\N$ satisfies Assumption \ref{a:abstract-pf} for functions $\funlatenthelp\in\{1,\funlatentmain\}$, with
  $  \gammaps_\nrt^{(0)}(\gss_n^{(0)} \funlatenthelp)$ and
  $  \gammaps_\nrt^{(\infty)}(\gss_n^{(\infty)} \funlatenthelp)$ finite,
  and
  $\gammaps^{(\infty)}(\gss_\nrt^{(\infty)})>0$.
  For each $k\in \{1{:}\nrii\}$, if one runs independently Algorithm \ref{alg:debiased-pf}, forming $\zeta_k(\funlatenthelp)$ from the output as in \ref{eq:debiased-pf} for each $k$, then
  $$
  E_{\nrii,\nrp,\mathbf{p}}(\funlatentmain)
  \defeq
 \frac{\sum_{k=1}^\nrii \zeta_k(\funlatentmain)}
      {\sum_{k=1}^\nrii \zeta_k(1)}
      \xrightarrow{\nrii\to\infty}
      \hat{\etaps}_n^{(\infty)}(\funlatentmain)
      \qquad
      \text{almost surely.}
   $$
Moreover, with $\bar{\funlatentmain}\defeq \funlatentmain - \hat{\etaps}_n^{(\infty)}(\funlatentmain)$,
$$
\sqrt{\nrii}[E_{\nrii,\nrp,\mathbf{p}}(\funlatentmain) - \hat{\etaps}_n^{(\infty)}(\funlatentmain)]
    \xrightarrow{\nrii\to\infty}
    \mathcal{N}(0,\sigma^2)
    \qquad
    \text{in distribution},
    $$
    where
    $$
    \sigma^2= \frac{ s_{\bar{\funlatentmain}} - \big( \gammaps^{(\infty)}(\gss_\nrt^{(\infty)} \bar{\funlatentmain}) - \gammaps^{(0)}(\gss_\nrt^{(0)} \bar{\funlatentmain})\big)^2}{ [\gammaps^{(\infty)}(\gss_\nrt^{(\infty)})]^2}.
    $$
\end{proposition}
The above result follows directly from the results of Section \ref{sec:method}.  It can also be seen as a multilevel version of Proposition 23 of \cite{vihola-helske-franks}, with straightforward estimators for $\sigma^2$. See Section \ref{sec:efficiency} for suggested choices for $\mathbf{p}$ and number of particles run at each level.

\section{A variance bound for the delta particle filter}
\label{sec:bound}
In this section we give theoretical results for the $\Delta$PF (Algorithm \ref{alg:delta-pf}) in the setting of HMM diffusions, which can be used to verify finite variance and therefore consistency of related estimators. In particular, Corollary \ref{theo:maincor} below can be used to verify Assumption \ref{a:abstract-pf}.
\subsection{Hidden Markov model diffusions}
We consider an HMM diffusion and corresponding Feynman-Kac model as in Section \ref{sec:delta-pf}.
We omit $\theta$ from the notation in the following, which is allowed as the remaining conditions and results in this Section \ref{sec:bound} will hold uniformly in $\thetaone$ (i.e.~any constants are independent of $\thetaone$).  The following will be assumed throughout.
\setcounter{assumptionx}{3}
\begin{assumptionx}\label{a:diffusion}
The coefficients $a^j, b^{j,k}$ are twice differentiable for $j,k= 1,\ldots, d$, and
\begin{quote}
\begin{itemize}
\item[(i)] {\bf uniform ellipticity}: $b(x)b(x)^T$ is uniformly positive definite;
\item[(ii)] {\bf globally Lipschitz}:
there is a $C>0$ such that 
$|a(x)-a(y)|+|b(x)-b(y)| \leq C |x-y|$ 
for all $x,y \in \mathbb{R}^d$; 
\end{itemize}
\end{quote}
\end{assumptionx}
Let $\msstrue(x, \ud y) \eqdef \msstrue_p(x,\ud y)$ for $p=0{:}\nrt$ denote the Markov transition of the unobserved diffusion \ref{eq:sde}, i.e.~the distribution of the solution $X_1$ of \ref{eq:sde} started at $X_0=x$.  With similar setup from Section \ref{sec:delta-pf}, with $\sss_{0:n}\defeq\mathsf{X}^{n+1}$, we have that \ref{eq:flow} takes the form
$$
\gammasstrue_n(\funlatentmain)
=
\int
\funlatentmain(x_{0:\nrt})
\Big(\prod_{p=0}^{\nrt-1} G_p(x_{p})\Big)
\eta_0(\ud x_0)
\prod_{p=1}^\nrt
\msstrue(x_{p-1}, \ud x_p). 
$$
In practice one usually must approximate the true dynamics $\msstrue(x,\ud y)$ of the underlying diffusion with a simpler transition $M^{(\ell)}(x,\ud y)$, based on some Euler type scheme using a discretisation parameter $h_{\ell}= 2^{-\ell}$ for $\ell\ge 0$; see \cite{kloeden-platen}.
The scheme allows for a coupling of the diffusions $(X_t^{(\ell)}, X_t^{(\ell-1)})_{t\ge 0}$ running at discretisation levels $\ell$ and $\ell-1$ (based on using the same Brownian path $W_t$),
such that for some $\beta\in\{1,2\}$, we have
\begin{equation}\label{eq:coup_h_cont}
\E_{(x,y)}[|X_1^{(\ell)}-X_1^{(\ell-1)}|^2] \leq M(|x-y|^2 + h_{\ell}^{\beta}),
\end{equation}
where $M<\infty$ does not depend on $\ell\ge 1$.
In particular, if the diffusion coefficient $b(X_t)$ in \ref{eq:sde} is constant or if a Milstein scheme can be applied otherwise, then $\beta=2$; otherwise $\beta=1$; see Proposition D.1 of \cite{mlpf}.

\subsection{Variance bound}
Assume we are in the above HMM diffusion setting, and that the coupling of Assumption \ref{a:coupled-fk} holds,
with symbols $(F,C)$ equal to $(\ell,\ell-1)$ for $\ell\ge 1$, and $G_p^{(\ell)}=G_p^{(\ell-1)}\defeq G_p$ for $p=0{:}\nrt$.
Running Algorithm \ref{alg:delta-pf}, we recall that $\Delta_\ell(\funlatentmain)$, defined in \ref{eq:Delta}, satisfies, by Proposition \ref{prop:delta-pf},
$$
\E[\Delta_\ell(\funlatentmain)]
=
\gammass_n^{(\ell)}(G_n \funlatentmain) - \gammass_n^{(\ell-1)}(G_n \funlatentmain),
$$
regardless of the number $\nrp\ge 1$ of particles.

A (measurable) function $\funlatentmain:\mathsf{X}\to \R$ is Lipschitz, denoted $\funlatentmain \in\textrm{Lip}(\mathsf{X})$, if for some $C'<\infty$, $|\funlatentmain(x) - \funlatentmain(y)| \le C' |x-y|$ for all $x,y\in\mathsf{X}$.    
\setcounter{assumptionx}{0} 
\begin{assumptionx}\label{a:fk}
  The following conditions hold for the model $(\mss_n,\gss_n)$:
    \begin{quote}
  \begin{hypA}\label{hyp:1}
\begin{enumerate}
\item[(i)]{$\|G_n\|<\infty$ for each $n\geq 0$.}
\item[(ii)]{$G_n\in\textrm{Lip}(\mathsf{X})$ for each $n\geq 0$.}
\item[(iii)]{$\inf_{x\in\mathsf{X}}G_n(x)>0$ for each $n\geq 0$.}
\end{enumerate}
\end{hypA}
  \begin{hypA}\label{hyp:2}
    For every $n\geq 1$, $\varphi\in \textrm{Lip}(\mathsf{X})\cap\mathcal{B}_b(\mathsf{X})$
    there exist a $C'<\infty$ such that for $s\in\{F,C\}$, we have for every $(x,y)\in\mathsf{X}\times\mathsf{X}$ that
    $
    |M_n^s(\varphi)(x)-M_n^s(\varphi)(y)|\leq C'|x-y|.
    $
\end{hypA}
\end{quote}
\end{assumptionx}
In the following results for $\Delta_\ell(\funlatentmain)$, the constant $M<\infty$ may change
from line-to-line. It will not depend upon $\nrp$ or $\ell$ (or $\theta$), but may depend on the time-horizon $\nrt$ or the function $\funlatentmain$. 
$\mathbb{E}$ denotes expectation w.r.t.~the law associated to the $\Delta$PF started at $(\xss,\xss)$, with $\xss\in \mathsf{X}$.
Below we only consider multinomial resampling in the $\Delta$PF for simplicity, though Theorem \ref{theo:mainthm} and Corollary \ref{theo:maincor} can be proved also assuming other resampling schemes. 

\begin{theorem}\label{theo:mainthm}
Assume (A\ref{hyp:1}-\ref{hyp:2}).  Then for any $\varphi\in\mathcal{B}_b(\mathsf{X}^{\nrt+1})\cap\textrm{\emph{Lip}}(\mathsf{X}^{\nrt+1})$, 
there exists a $M<\infty$ such that
$$
\mathbb{E}\Big[\Big(\Delta_\ell(\varphi) - \mathbb{E}[\Delta_\ell(\varphi)]\Big)^2\Big]
\leq \frac{Mh_{\ell}^{2\wedge\beta}}{\nrp},
\qquad
\text{with $\beta$ as in \ref{eq:coup_h_cont}}.
$$
\end{theorem}

\begin{corollary}\label{theo:maincor}
  Assume (A\ref{hyp:1}-\ref{hyp:2}).  Then for any $\varphi\in\mathcal{B}_b(\mathsf{X}^{\nrt+1})\cap\textrm{\emph{Lip}}(\mathsf{X}^{\nrt+1})$,
there exists a $M<\infty$ such that
$$
\mathbb{E}\Big[\Big(\Delta_\ell(\varphi)\Big)^2\Big]
\leq M\Big(\frac{h_{\ell}^{2\wedge\beta}}{\nrp} +  h_{\ell}^2 \Big),
\qquad\qquad
\text{with $\beta$ as in \ref{eq:coup_h_cont}}.
$$
\end{corollary}
The proofs are given in Appendix \ref{app:bound}.

Based on Corollary \ref{theo:maincor}, Recommendation \ref{rec:allocations} of Section \ref{sec:efficiency} suggests allocations for $\mathbf{p}$ and $\nrp_\ell$ in the $\Delta$PF (Algorithm \ref{alg:delta-pf}) to optimally use resources and minimise variance \ref{eq:variance}.

\section{Unbiased joint inference for hidden Markov model diffusions}
\label{sec:method} 

We are interested in unbiased inference for the Bayesian model posterior
$$
    \pi^{(\infty)}(\ud \theta, \ud x_{0:\nrt})
    \propto \mathrm{pr}(\ud \theta) G_{\nrt}^{(\theta)}(x_{\nrt}) \gamma_{\nrt}^{(\theta,\infty)} (\ud x_{0:\nrt}),
$$
    where $\mathrm{pr}(\ud \theta) = \mathrm{pr}(\theta) \ud \theta$ is the prior on the model parameters, and 
    $$
    \gamma_{\nrt}^{(\theta,\infty)}(\ud x_{0:\nrt}) =
    \Big(\prod_{t=0}^{\nrt-1} G_t^{(\theta)}(x_{t})
    \Big)
    \eta_0^{(\theta)}(\ud x_0)
    \prod_{t=1}^\nrt
    M_t^{(\theta,\infty)}(x_{t-1},\ud x_t).
    $$    
Here, $M^{(\theta,\infty)}_t$
corresponds to the transition density of the
diffusion model of interest.
The dependence of the HMM on $\theta$ is made explicit in this section.  As in Section \ref{sec:bound}, we assume the transition densities 
$M_t^{(\theta,\infty)}$ cannot be simulated, 
but that there are increasingly refined discretisations $M_t^{(\theta,\ell)}$ approximating $M_t^{(\theta,\infty)}$ in the sense of \ref{eq:refined} (with $E_{0:\nrt}\defeq \mathsf{X}^{\nrt+1}$).  

\subsection{Randomised MLMC IS type estimator based on coarse-model PMMH}
We now consider Algorithm \ref{alg:is-mlmc} for joint inference w.r.t.~the above Bayesian posterior.  Algorithm \ref{alg:is-mlmc} uses the following ingredients:
\begin{enumerate}[(i)]
    \item 
$\check{\mss}^{(\theta,\ell)}_{0:\nrt}$ satisfying Assumption
\ref{a:coupled-fk}\ref{a:coupled-fk-kernel} with $M_{0:\nrt}^F = M_{0:\nrt}^{(\theta,\ell)}$, 
and $M_{0:\nrt}^C = M_{0:\nrt}^{(\theta,\ell-1)}$.
\item 
$\check{G}^{(\theta)}_{0:\nrt}$ defined as in Assumption 
\ref{a:coupled-fk}\ref{a:coupled-fk-potential}, with $G_{0:\nrt}^F = G_{0:\nrt}^C =
G_{0:\nrt}^{(\theta)}$.
\item Metropolis-Hastings proposal distribution $q(\uarg\mid \theta)$ 
  for parameters.
\item Algorithm constant $\epsilon\ge 0$ (e.g.~$\epsilon=10^{-10}$; see Remark \ref{alg:is-mlmc-remark}\ref{alg:is-mlmc-remark-pmmh} below).
\item Number of MCMC iterations $\nri\in\N$ and number of particles $\nrp\in\N$.
\item Probability mass $\mathbf{p} = (p_\ell)_{\ell\in\N}$ on $\N$ with $p_\ell>0$ for all
  $\ell\in\N$.
\end{enumerate}


\begin{algorithm}[t]
\caption{Randomised multilevel importance sampling type estimator.}
\label{alg:is-mlmc} 
\raggedright
\textbf{Input:} $\Big((M_{0:\nrt}^{(\theta,\ell)}, G_{0:\nrt}^{(\theta)})\Big)_{(l,\theta)\in(\{0\}\times\mathbb{N})\times \mathsf{T}}$, prior $\mathrm{pr}(\ud \theta)$, $\nrp$ the number of particles, $\nri$ the number of iterations, $q(\uarg\mid \theta)$ a proposal density for Metropolis-Hastings, 
 $\epsilon\geq 0$, probability $\mathbf{p}=(p_\ell)_{\ell\in\N}$ and $(\Theta_0,V_0^{(1:\nrp)},\mathbf{X}_0^{(1:\nrp)})$ such
  that $\sum_{i=1}^{\nrp} V_0^{(i)}>0$.
\begin{enumerate}[({P}1)]
\item 
  \label{item:approx-phase}
  For $k=1{:}\nri$, iterate:
  \begin{enumerate}[(i)]
      \item Propose $\hat{\Theta}_k\sim q(\uarg \mid \Theta_{k-1})$.
      \item Run Algorithm \ref{alg:pf} with
        $(M_{0:\nrt}^{(\hat{\Theta}_k,0)},G_{0:\nrt}^{(\hat{\Theta}_k)},\nrp)$ and call the output
        $(\hat{V}_k^{(1:\nrp)},
        \hat{\mathbf{X}}_k^{(1:\nrp)})$.
      \item With probability
        \[
            \min\bigg\{1, \frac{\mathrm{pr}(\hat{\Theta}_k)
                q(\Theta_{k-1}\mid \hat{\Theta}_k) 
                \big(\sum_{i=1}^{\nrp}\hat{V}_k^{(i)} + \epsilon\big)}{
                \mathrm{pr}(\Theta_{k-1})q(\hat{\Theta}_k\mid
                \Theta_{k-1})
                \big(\sum_{j=1}^{\nrp} V_{k-1}^{(j)} + \epsilon\big)}
              \bigg\},
        \]
        accept and set 
        $(\Theta_k,V_k^{(1:\nrp)},\mathbf{X}_k^{(1:\nrp)}) \gets
        (\hat{\Theta}_k,\hat{V}_k^{(1:\nrp)}, \hat{\mathbf{X}}_k^{(1:\nrp)})$; otherwise
        set $(\Theta_k,V_k^{(1:\nrp)},\mathbf{X}_k^{(1:\nrp)}) \gets
        (\Theta_{k-1},V_{k-1}^{(1:\nrp)},\mathbf{X}_{k-1}^{(1:\nrp)})$.
  \end{enumerate}
\item 
  \label{item:corr-phase}
  For every $k\in \{1{:}\nri\}$, independently, conditional on 
  $(\Theta_k,V_k^{(1:\nrp)},\mathbf{X}_k^{(1:\nrp)})$:
  \begin{enumerate}[(i)]
  \item Set
    $\mathbf{X}^{(1:\nrp)}_{k,0} \defeq
    \mathbf{X}^{(1:\nrp)}_{k}$,
    and set 
        $W_{k,0}^{(i)} \defeq
   V_{k}^{(i)} / \big(\sum_{j=1}^{\nrp} V_k^{(j)}+\epsilon\big)$.
      \item Sample $L_k \sim \mathbf{p}$ independently from the other random variables.
      \item Run the $\Delta$PF (Algorithm \ref{alg:delta-pf}) with 
        $(\check{\mss}_{0:\nrt}^{(\Theta_k,L_k)}$,
        $\check{G}_{0:\nrt}^{(\Theta_k)}, \nrp)$, and
        call the output
        $(V_{k,L_k}^{(1:2\nrp)}, \mathbf{X}_{k,L_k}^{(1:2\nrp)})$.
        Set $W_{k,L_k}^{(i)} \defeq
        V_{k,L_k}^{(i)}/\big[p_{L_k}\big(\sum_{j=1}^{\nrp}
          V_k^{(j)}+\epsilon\big)\big]$.
  \end{enumerate}
   \end{enumerate}
  \raggedright
  \textbf{Output:} 
$$
    E_{\nri,\nrp,\mathbf{p}}(\funfullmain) \defeq \frac{\sum_{k=1}^\nri \big[
      \sum_{i=1}^{\nrp}
      W_{k,0}^{(i)}\funfullmain(\Theta_k, \mathbf{X}_{k,0}^{(i)}) 
      + \sum_{i=1}^{2\nrp}
      W_{k,L_k}^{(i)} \funfullmain(\Theta_k, \mathbf{X}_{k,L_k}^{(i)})
      \big]}{\sum_{k=1}^\nri \big[\sum_{i=1}^{\nrp}
      W_{k,0}^{(i)} +
      \sum_{i=1}^{2\nrp} W_{k,L_k}^{(i)}\big]}.
    $$
\end{algorithm} 

\begin{remark}\label{alg:is-mlmc-remark}
Before stating consistency and central limit theorems, we briefly
discuss various aspects of this approach, which are appealing from
a practical perspective, and we also mention certain algorithmic
modifications which could be further considered.
  \begin{quote}
\begin{enumerate}[(i)]
\item \label{alg:is-mlmc-remark-pmmh}
  The first phase (P\ref{item:approx-phase}) of Algorithm 
\ref{alg:is-mlmc} implements a PMMH type algorithm
\cite{andrieu-doucet-holenstein}. If $\epsilon=0$, this is exactly PMMH
targeting the model
$\pi^{(0)}(\ud \theta, \ud x_{0:\nrt})\propto \mathrm{pr}(\ud \theta) G_\nrt^{(\theta)}(x_\nrt) \gamma_\nrt^{(\theta,0)}(\ud x_{0:\nrt})$. It is generally safer to choose
$\epsilon>0$ \cite{vihola-helske-franks}, which ensures that the IS type
correction 
in phase (P\ref{item:corr-phase}) will yield consistent inference for the ideal model 
$$\pi^{(\infty)}(\ud \theta, \ud x_{0:\nrt})\propto \mathrm{pr}(\ud \theta) G_\nrt^{(\theta)}(x_\nrt)
\gamma_\nrt^{(\theta,\infty)}(x_{0:\nrt})$$ (Theorem \ref{thm:consistency}). Setting $\epsilon>0$ may be helpful otherwise in terms of improved mixing, as the PMMH will target marginally an averaged probability between a `flat' prior and a `multimodal' $\ell=0$ marginal posterior.
\item It is only necessary to implement PMMH for the coarsest level.
  This is typically relatively cheap, and therefore allows for 
  a relatively long MCMC run. Consequently, relative cost of burn-in is
  small, and if the proposal $q$ is adapted (see \cite{andrieu-thoms}), it has time to converge.
\item The (potentially costly) r$\Delta$PFs 
  are applied independently for each $\Theta_k$, which allows for 
  efficient parallelisation.
\item We suggest that the number of particles $N$, here referred to as `$\nrp_0$', used in the PMMH be chosen based on \cite{doucet-pitt-deligiannidis-kohn,sherlock-thiery-roberts-rosenthal}, while the number of particles `$N_{\ell}$' (and $p_\ell$) can be optimised for each level $\ell$ based on Recommendation \ref{rec:allocations} of Section \ref{sec:efficiency}, or kept constant.  One can also afford to increase the number of particles when a `jump chain' representation is used (see the following remark).
\item The r$\Delta$PF corrections may be calculated only once for each
  accepted state \cite{vihola-helske-franks}. That is, 
  suppose $(\tilde{\Theta}_k, \tilde{V}_k^{(1:\nrp)}, 
  \tilde{\mathbf{X}}_k^{(1:\nrp)})_{k=1}^{\nrijump}$ are the accepted
  states, $(D_k)_{k=1}^{\nrijump}$ are the corresponding holding
  times, and $(\tilde{V}_{k,L_k}^{(1:2\nrp_{L_k)}}, 
  \tilde{\mathbf{X}}_{k,L_k}^{(1:2\nrp_{L_k})})_{k=1}^{\nrijump}$ are corresponding $\Delta$PF outputs, then the estimator is formed as in Algorithm \ref{alg:is-mlmc} using $(\tilde{\Theta}_k, \tilde{V}_k^{(1:\nrp)}, 
  \tilde{\mathbf{X}}_k^{(1:\nrp)})$, and accounting for the holding times
  in the weights defined as
  $W_{k}^{(i)} \defeq
   D_k \tilde{V}_{k}^{(i)} / \big(\sum_{j=1}^{\nrp} \tilde{V}_k^{(j)}+\epsilon\big)$
   and $W_{k,L_k}^{(i)} \defeq
        \tilde{V}_{k,L_k}^{(i)}/\big[p_{L_k}\big(\sum_{j=1}^{\nrp}
          \tilde{V}_k^{(j)}+\epsilon\big)\big]$.
\item In case the Markov chain in (P\ref{item:approx-phase}) phase is slow
  mixing, (further) thinning may be applied (to the jump chain) before the (P\ref{item:corr-phase}) phase.
\item In practice, Algorithm \ref{alg:is-mlmc} may be implemented in
  an on-line fashion w.r.t.~the number of
  iterations $\nri$, and by progressively refining the estimator
  $E_{\nri,\nrp,\mathbf{p}}(\funfullmain)$. The r$\Delta$PF corrections
  may be calculated in parallel with the Markov chain.
\item In Algorithm \ref{alg:is-mlmc}, the r$\Delta$PFs
  depend only on $\Theta_k$. They could depend also on 
  $V_{k}^{(i)}$ and $\mathbf{X}_k^{(i)}$, but it is not clear how such
  dependence could be used in practice to achieve better performance.
  Likewise, the `zeroth level' estimate in Algorithm \ref{alg:is-mlmc}
  is based solely on particles in (P\ref{item:approx-phase}), but it
  could also be based on (additional) new particle filter output.
\item In order to save memory, it is possible also to `subsample' only one 
  trajectory $\mathbf{X}_{k}^*$ from $\mathbf{X}_{k}^{(1:\nrp)}$, such that
  $\P[\mathbf{X}_{k}^* = \mathbf{X}_{k}^{(i)}]\propto V_k^{(i)}$, and set $W_{k,0}^* \defeq
  \sum_{i=1}^\nrp W_{k,0}^{(i)}$, and similarly in Algorithm
  \ref{alg:delta-pf} find $\check{\mathbf{X}}^*$ such that
  $\P[\check{\mathbf{X}}^* = \check{\mathbf{X}}^{(i)}]\propto \check{V}^{(i)}$, setting $\mathbf{X}_{k,L_k}^{*(1:2)}\defeq \check{\mathbf{X}}^*$, and defining from the usual output of Algorithm \ref{alg:delta-pf}, $W^{*(1)}_{k,L_k} \defeq \sum_{i=1}^\nrp W_{k,L_k}^{(i)}$ and $W^{*(2)}_{k,L_k} \defeq \sum_{i=\nrp+1}^{2\nrp} W_{k,L_k}^{(i)}$.  The subsampling output estimator then takes the form,
  \[
      E_{\nri,\nrp,\mathbf{p}}^{\text{subsample}}(\funfullmain)
      \defeq \frac{\sum_{k=1}^\nri \big[W_{k,0}^* \funfullmain(\Theta_k, \mathbf{X}_{k}^*) +
        \sum_{i=1}^2 W_{k,L_k}^{*(i)} \funfullmain(\Theta_k,
        \mathbf{X}_{k,L_k}^{*(i)}) \big]
        }{\sum_{k=1}^\nri \big[W_{k,0}^* + 
         \sum_{i=1}^2 W_{k,L_k}^{*(i)}\big] 
        }.
  \]
  Note, however, that the asymptotic variance of this estimator is higher,  because \\
  $E_{\nri,\nrp,\mathbf{p}}(\funfullmain)$ may be viewed as a Rao-Blackwellised version
  of $E_{\nri,\nrp,\mathbf{p}}^{\text{subsample}}(\funfullmain)$.
\end{enumerate}
\end{quote}
\end{remark}

\subsection{Consistency and central limit theorem}

\begin{theorem} 
    \label{thm:consistency} 
Assume that the algorithm constant $\epsilon\ge 0$ is chosen positive, and that the Markov chain 
$(\Theta_k,
X_k^{(1:\nrp)}, V_k^{(1:\nrp))})_{k\ge 1}$ is $\psi$-irreducible, and that
$\pi^{(0)}(\funfullmain)$ and $\pi^{(\infty)}(\funfullmain)$ are finite.
For each $\theta\in \stheta$, suppose Assumption
\ref{a:abstract-pf} holds for  
$\funfullhelp \equiv 1$ and $\funfullhelp= \funfullmain^{(\theta)} \defeq \funfullmain(\theta,\uarg)$, with
$\mss_{0:\nrt}^{(\ell)}\defeq \mss_{0:\nrt}^{(\theta,\ell)}$ and $\gss_{0:\nrt}^{(\ell)}\defeq \gss_{0:\nrt}^{(\theta)}$.
Assume
\[
    \int \mathrm{pr}(\theta) \big(\sqrt{s_{1}(\theta)} +
    \sqrt{\smash{s_{\funfullmain^{(\theta)}}(\theta)}\vphantom{()}}
    \big)\ud \theta  <\infty.
\]
Then, the estimator of Algorithm
\ref{alg:is-mlmc} is strongly consistent:
\[
    E_{\nri,\nrp,\mathbf{p}}(\funfullmain) \xrightarrow{\nri\to\infty}
    \int \pi^{(\infty)}(\ud \theta, \ud x_{0:\nrt}) \funfullmain(\theta, x_{0:\nrt})
    \qquad\text{(a.s.)}
\]
\end{theorem} 

\begin{remark} 
Regarding Theorem \ref{thm:consistency}, whose proof is given in Appendix \ref{app:method}:
  \begin{quote}
    \begin{enumerate}[(i)]

\item If all potentials $G_t$ are strictly positive, the
  algorithm constant $\epsilon$ may be taken to be zero. However, if
  $\epsilon=0$ and Algorithm
  \ref{alg:pf} with 
  $(M_{0:\nrt}^{(\hat{\Theta}_k,0)}, G_{0:\nrt}^{(\hat{\Theta}_k)},\nrp)$  
  can produce an estimate with $\sum_{i=1}^{\nrp} V^{(i)}=0$ with
  positive probability, the consistency may be lost
  \cite{vihola-helske-franks}.
\end{enumerate}
\end{quote}
\end{remark} 


\begin{proposition}\label{prop:clt} 
  Suppose that the conditions of Theorem \ref{thm:consistency} hold.  Suppose additionally that $\pi^{(\infty)}(\funfullmain^2) <\infty$ and that the base chain $(\Theta_k, V^{(1:\nrp)}_k,\mathbf{X}^{(1:\nrp)}_k)_{k\ge 1}$ is aperiodic, with transition probability denoted by $P$.  Then,
\[
    \sqrt{\nri}\big[ E_{\nri,\nrp,\mathbf{p}}(\funfullmain) - \pi^{(\infty)}(\funfullmain) \big]
    \xrightarrow{\nri\to\infty} \mathcal{N}(0,\sigma^2),
    \qquad \text{in distribution},
\]
whenever the asymptotic variance 
\begin{equation}\label{eq:asvar}
    \sigma^2 = \frac{\var(P, \mu_{\bar{\funfullmain}}) +
      \pimcmc(\sigma^2_\xi)}{c^2},
\end{equation}
is finite.  Here, $\bar{\funfullmain} \defeq \funfullmain - \pi^{(\infty)}(\funfullmain)$, $c>0$ is a constant (equal to $\pimcmc(\mu_1))$, and 
\begin{align*}
    \sigma^2_\xi(\theta, v^{(1:\nrp)}, \mathbf{x}^{(1:\nrp)}) &\defeq 
    \var\big( \xi_k(\bar{\funfullmain})\bigmid \Theta_k=\theta,V_k^{(1:\nrp)}=v^{(1:\nrp)},
    \mathbf{X}_k^{(1:\nrp)}=\mathbf{x}^{(1:\nrp)}\big) \\
   & = \frac{s_{\bar{\funfullmain}^{(\theta)}}(\theta) -
      \big(\gammass_\nrt^{(\theta,\infty)}(G_n \bar{\funfullmain}^{(\theta)} ) - 
\gammass_\nrt^{(\theta,0)}(G_n \bar{\funfullmain}^{(\theta)})  \big)^2
      }{\big(\sum_{i=1}^{\nrp} v^{(i)} + \epsilon\big)^2}.
\end{align*}
\end{proposition} 

\begin{remark}
Proposition \ref{prop:clt} follows from Theorem 7 \cite{vihola-helske-franks}.
  We suggest that $\nrp=\nrp_0$ for (P\ref{item:approx-phase}) be chosen based on
\cite{doucet-pitt-deligiannidis-kohn,sherlock-thiery-roberts-rosenthal} to minimise $\var(P, \mu_{\bar{\funfullmain}})$, and that $(p_\ell)$ and $\nrp=\nrp_\ell$ in (P\ref{item:corr-phase}) for the r$\Delta$PF be chosen as in Recommendation \ref{rec:allocations} of Section \ref{sec:efficiency}, to minimise $\sigma_\xi^2$, subject to cost constraints, in order to jointly minimise $\sigma^2$.  However, the question of the optimal choice for $N_0$ in the IS context is not yet settled.
\end{remark}
  
\begin{remark} Regarding $\psi$-irreducibility and aperiodicity in Theorem \ref{thm:consistency} and Proposition \ref{prop:clt}, these are inherited by the coarse PMMH chain \cite{andrieu-doucet-holenstein} if the corresponding idealised marginal Metropolis-Hastings for the coarse model has these properties.  
\end{remark}

\section{Asymptotic efficiency and randomised multilevel considerations} 
\label{sec:efficiency}
We summarise the results of this section by suggesting the following safe allocations for probability $\mathbf{p}=(p_\ell)_{\ell\in \N}$ and number $\nrp=\nrp_\ell$ of particles at level $\ell\ge 1$ in the $\Delta$PF (Algorithm \ref{alg:delta-pf}) used in Algorithm \ref{alg:debiased-pf} and Algorithm \ref{alg:is-mlmc}, and Proposition \ref{prop:independent-pf}, with $\beta$ given in \ref{eq:coup_h_cont} in the HMM diffusion context of Section \ref{sec:bound}, or, indeed, with $\beta$ given in the abstract framework under Assumption \ref{a:rates} given later. See also Figure \ref{fig:optimal} for the recommended allocations.
\begin{recommendation}\label{rec:allocations}
With strong error convergence rate  $\beta$ given in \ref{eq:coup_h_cont}, we suggest the following for $\mathbf{p}=(p_\ell)_{\ell \in \N}$ and $\nrp_\ell\in \N$ in $\Delta$PF (Algorithm \ref{alg:delta-pf}):
  \begin{quote}
  \begin{enumerate}[($\beta=\,$1)]
  \item (e.g. Euler scheme).
Choose $p_\ell = (\frac{1}{2})^{\ell}$ and $N_\ell \propto 1$ constant.
  \item (e.g. Milstein scheme).
    Choose $p_\ell \propto 2^{-1.5 \ell}\approx (\frac{1}{3})^\ell$ and $\nrp_\ell\propto 1$ constant. 
  \end{enumerate}
  \end{quote}
\end{recommendation}
The suggestions are based on Corollary \ref{theo:maincor} of Section \ref{sec:bound}, and 
Proposition \ref{prop:efficiency-both-finite} ($\beta=2$) and \ref{prop:subcanonical-optimised} ($\beta=1$) given below (with weak convergence rate $\alpha=1$; see Figure \ref{fig:optimal} for general $\alpha$).  In the Euler case, although the theory below gives the same computational complexity order by choosing any $\rho \in [0,1]$ and setting $p_\ell\propto 2^{-(1+\rho)\ell}$ and $N_\ell\propto 2^{\rho \ell}$, the experiment in Section \ref{sec:experiments} gave a better result using simply $\rho=0$, corresponding to no scaling.  However, this may depend on the implementation.

\subsection{Efficiency framework}
The asymptotic efficiency of simulation-based estimators has been considered theoretically in \cite{glynn-whitt}; see \cite{giles-or} in the dMLMC context.  The developments of this section follow \cite{rhee-glynn} for rMLMC (originally in the i.i.d.~setting without observations), while also giving some extensions (also applicable to that setting).  We will see that the basic rMLMC results carry over to our setting involving MCMC and randomised estimators based on PF outputs, but also discover a novelty in the common Euler case ($\beta=1$ in Figure \ref{fig:optimal}).  Proofs are given in Appendix \ref{app:efficiency}.

We are interested in modeling the computational costs involved in running Algorithm \ref{alg:is-mlmc}; the algorithm of Proposition \ref{prop:independent-pf} is recovered with $\stheta=\{\theta\}$.
Let $\tau_{\Theta_k,L_k}$ represent the combined \emph{cost at iteration $k$} of the base Markov chain and weight calculation in Algorithm \ref{alg:is-mlmc}, so that the
\emph{total cost} $\cost(\nrii)$ of Algorithm \ref{alg:is-mlmc} with $\nrii$ iterations is
  $$
  \cost(\nrii)\defeq
  \sum_{k=1}^\nrii  \tau_{\Theta_k,L_k}.
  $$
The following assumption seems natural in our setting.
\begin{assumption}\label{a:costs}
  For $\Theta_k\in \mathsf{T}$, a family $\{\tau_{\Theta_k,\ell} \}_{k, \ell \ge 1}$ consists of positive-valued random variables that are independent of $\{L_k\}_{k\ge 1}$, where $L_k\sim \mathbf{p}$ i.i.d., and that are conditionally independent given $\{\Theta_k\}_{k\ge 1}$, such that $\tau_{\Theta_k,\ell}$ depends only on $\Theta_k\in\mathsf{T}$ and $\ell\in \mathbb{N}$.  
\end{assumption}
  Under a budget constraint $\kappa>0$, the \emph{realised length of the chain} is $\lmcmc(\kappa)$ iterations, where
  \begin{equation*}
  \lmcmc(\kappa)
  \defeq
  \max\{
  \nrii\ge 1: \cost(\nrii) \le \kappa
  \}.
\end{equation*}
Under a budget constraint, the CLT of Proposition \ref{prop:clt} takes the following altered form, where here $\pimcmc_m(\ud \theta)$ denotes the $\theta$-marginal of the invariant probability measure (given as \ref{eq:pmmh-invar} in Appendix \ref{app:method})  of the base Markov chain (equal to the $\theta$-marginal posterior of the $\ell=0$ model).
\begin{proposition}\label{prop:is-efficiency}
  If the assumptions of Proposition \ref{prop:clt} hold with $\sigma^2<\infty$,
 and if $\mIRE\defeq \E_{\pimcmc_m\otimes \mathbf{p}}[\taufun]<\infty$ with $\taufun(\theta,\ell)\defeq \E[\tau_{\Theta_k,L_k}| \Theta_k =\theta, L_k =\ell]$,
  then
  \begin{equation}\label{eq:is-efficiency}
  \sqrt{\kappa}
  \big[ E_{\lmcmc(\kappa),\nrp,\mathbf{p}}(\funfullmain) - \pi^{(\infty)}(\funfullmain)\big]
    \xrightarrow{\kappa\to\infty}\mathcal{N}(0,\E[\taufun]\sigma^2),
    \qquad \text{in distribution}.
  \end{equation}
  \end{proposition}
\begin{remark} The quantity $\E[\taufun]\sigma^2$ is called the `inverse relative efficiency' by \cite{glynn-whitt}, and is considered a more accurate quantity than the asymptotic variance ($\sigma^2$ here) for comparison of Monte Carlo algorithms run on the same computer, as it takes into account also the average computational time.  
  \end{remark}
In the following we consider (possibly) variance reduced (if $\rate>0$) versions of $\Delta_\ell(\funfullhelp)$ of Assumption \ref{a:abstract-pf}, denoted $\Delta_\ell$, where $\funfullhelp=\funfullmain^{(\theta)}$, based on running the $\Delta$PF (Algorithm \ref{alg:delta-pf}) with parameters $\theta$, $\ell$ fixed.  The constant $C<\infty$ may change line-to-line, but does not depend on $\nrp$, $\ell$, or $\theta$, but may depend on the time-horizon $\nrt$ and function $\funfullmain$.          
  \begin{assumption}\label{a:rates}
Assumption \ref{a:costs} holds, and constants $2 \alpha\ge \beta>0$, $\costrate>0$, and $\rate \ge 0$ are such that the following hold:
  \begin{quote}
    \begin{enumerate}[(i)]
    \item\label{a:rates-cost}(Mean cost)
$\E[ \tau_{\theta,\ell}] \le C 2^{\costrate \ell(1+\rate)}$ 
    \item \label{a:rates-strong}(Strong order)
      $\E[\Delta_\ell^2] \le C 2^{-\ell(\beta + \rate)}  + C 2^{-2 \alpha\ell}$
    \item \label{a:rates-weak}(Weak order)
$|\E\Delta_\ell| \le C 2^{-\alpha \ell}$
    \end{enumerate}
    \end{quote}
  \end{assumption}
  \begin{remark} Regarding Assumption \ref{a:rates}: 
    \begin{quote}
    \begin{enumerate}[(i)]
    \item We only assume bounded mean cost in Assumption \ref{a:rates-cost}, rather than the almost sure cost bound commonly used.  This generalisation allows for the setting where occasional algorithmic runs may take a long time.
    \item In the original MLMC setting, the cost scaling $\costrate$ in Assumption \ref{a:rates-cost} is taken to be $\costrate=1$ \cite{giles-or,rhee-glynn}.  However, in settings involving uncertainty quantification, and where the forward solver may involve non-sparse matrix inversions, often $\costrate\ge 1$ \cite{cliffe2011multilevel,jasra2018multi,mlpf}.
      \item We assume in Assumption \ref{a:rates-cost} that the mean cost to form $\Delta_\ell$ is bounded by the $\costrate$-scaled product of the number of samples or particles $\nrp_\ell$ times the number of Euler time steps $2^\ell + 2^{\ell-1}$ together with the $O(\nrp_\ell)$-resampling cost, where there are $\nrp_\ell \propto 2^{\rate \ell}$ particles at level $\ell$.  Here, we recall that the stratified, systematic, and residual resampling algorithms have $O(\nrp_\ell)$ cost, as does an improved implementation of multinomial resampling; see \cite{cappe,douc-cappe-moulines}.
      \item With $\rate=0$, by Jensen's inequality one sees why $\alpha \ge \beta/2$ can be assumed, and that Assumption \ref{a:rates-strong} becomes $\E\Delta_\ell^2 \le C 2^{-\ell \beta}$.  
    \item $\rate\ge 0$ in Assumption \ref{a:rates-cost} and \ref{a:rates-strong} corresponds to using an average of
      $\nrp_\ell\defeq \lceil 2^{\rate \ell}\rceil$
      i.i.d samples of $\Delta_\ell^{(1)}$, i.e.~$\Delta_\ell = \frac{1}{\nrp_\ell}\sum_{i=1}^{\nrp_\ell} \Delta_\ell^{(i)}$, or, of more present interest to us, to increasing the number of particles used in a PF by a factor of $\nrp_\ell$ instead of the default lower number.  The former leads to 
      $
\E\Delta_\ell^2 = \frac{1}{\nrp_\ell} \var(\Delta_\ell^{(1)}) + \E[\Delta_\ell^{(1)}]^2,
$
justifying Assumption \ref{a:rates-strong}, as does Corollary \ref{theo:maincor}, with $\beta\in\{1,2\}$ and $\alpha=1$, for the $\Delta$PF (Algorithm \ref{alg:delta-pf}) in the HMM diffusion context (Section \ref{sec:bound}).    
      \end{enumerate}
      \end{quote}
    \end{remark}
  \begin{figure}
  \centering
    \begin{tikzpicture}[scale=2]
      \draw[style = help lines] (0,0) -- (0,1);
      \draw[style = help lines] (0,0) -- (4,0);
      \draw[style = help lines] (2,0) -- (2,.05);
      \draw[style = help lines] (0,.5) -- (.05,.5) ;
      
            \draw[thick,blue] (0,1) -- (2,.5) node[midway,sloped, below] {$2\alpha-\beta$};
      \draw[thick,blue] (2,.5) -- (2,0) node[midway, right] {$\rho$};
      \draw[thick,blue] (2,0) -- (4,0);
      \draw[red] (0,1) -- (2,1);
      \draw[red] (2,1) -- (2,.5) node[midway, anchor=west] {$1+\rho$};
      \draw[red] (2,.5) -- (4,.75) node[midway, sloped, below] {$(1+\beta)/2$};
      \draw (0,1) node[anchor=east] {$2\alpha$};
      \draw (0,.5) node[anchor=east] {$2\alpha-1$};

            \draw (2,-.5) node {$\beta$};
            \draw (2,0) node[anchor=north] {$1$};
            \draw (4,0) node[anchor=north] {$2 \alpha$};
            \draw (0,0) node[anchor= north east] {$0$};
      \draw[thick,blue]   (4.5,.75) -- (5,.75) node[anchor=west] {$\rho=\rho(\beta)$};
            \draw[red] (4.5,.5) -- (5,.5) node[anchor=west] {$r=r(\beta,\rho)$};
    \end{tikzpicture}
    \caption{Recommendations for number of particles $N_\ell\propto 2^{\rho \ell}$ and probability $p_\ell\propto 2^{-r \ell}$.  Here, $\gamma=1$ always, and $\rho\in [0,2\alpha-1]$ when $\beta=1$ provides a line of choices with the same \emph{order} of computational complexity.  In our particular experiment in Section \ref{sec:experiments}, however, the simple choice $\rho=0$, corresponding to no scaling in the particles, will turn out to be optimal.}
  \label{fig:optimal}
  \end{figure}
\begin{proposition}\label{prop:efficiency-both-finite}
  Suppose Assumption \ref{a:rates} and the assumptions of Proposition \ref{prop:clt} hold, with $\var(P, \mu_{\bar{\funfullmain}})<\infty$.  If $p_\ell \propto 2^{-r \ell}$ for some $r\in \big(\costrate(1+\rate), \min(\beta + \rate, 2\alpha)\big)$, then \ref{eq:is-efficiency} holds, i.e.
    \begin{equation*}
  \sqrt{\kappa}
  \big[ E_{\lmcmc(\kappa),\nrp,\mathbf{p}}(\funfullmain) - \pi^{(\infty)}(\funfullmain)\big]
    \xrightarrow{\kappa\to\infty}\mathcal{N}(0,\mIRE\sigma^2),
    \qquad \text{in distribution}.
  \end{equation*}
\end{proposition}
\begin{remark} Regarding Proposition \ref{prop:efficiency-both-finite}, in the common case $\costrate=1$ for simplicity:
  \begin{quote}
  \begin{enumerate}[(i)]
    \item
      If $\beta>1$ (`canonical convergence regime') and $\rate=0$, then a choice for $r\in (1,\beta)$ exists.  See also Theorem 4 of \cite{rhee-glynn} for a discussion of the theoretically optimal $\mathbf{p}$. 
\item
      If $\beta\le 1$ (`subcanonical convergence regime'), then $\beta + \rate \le 1 + \rate$ and so no choice for $r$ exists.
  \end{enumerate}
  \end{quote}
  \end{remark}

\subsection{Subcanonical convergence}
When $\beta>1$, within the framework above we have seen that a canonical convergence rate holds (Proposition \ref{prop:efficiency-both-finite}) because $\mIRE<\infty$ and $\sigma^2<\infty$.  When $\beta\le 1$, this is no longer the case, and one must choose between a finite asymptotic variance and infinite expected cost, or vice versa.  Assuming the former, and that a CLT holds (Proposition \ref{prop:clt}), for $\epsilon>0$ and $0<\delta <1$ the Chebyshev inequality implies that the number of iterations of Algorithm \ref{alg:is-mlmc} so that
  \begin{equation}\label{eq:chebyshev}
\P[ \abs{ E_{\nrii,\nrp,\mathbf{p}}(\funfullmain) -\pi^{(\infty)}(\funfullmain)} \le \epsilon]\ge 1-\delta,
  \end{equation}
  holds implies that $\nrii$ must be of the order $O(\epsilon^{-2})$.  The question is then how to minimise the total cost $\mathscr{C}(\nrii)$, or \emph{computational complexity}, involved in obtaining the $\nrii$ samples.  This will involve optimising for $(p_\ell)$ and $\nrp_\ell$ to minimise $\mathscr{C}(\nrii)$, while keeping the asymptotic variance finite.

  \begin{proposition}\label{prop:subcanonical-abstract}
    Suppose that the assumptions of Proposition \ref{prop:clt} hold with $\sigma^2<\infty$, and Assumption \ref{a:rates} holds with
    $
    \E[\tau_{\Theta_{k_0},L_{k_0} }]
    =\infty
    $
    for some $k_0\ge 1$.  If
    $$
    \sum_{k\ge 1} \sup_{j\ge 1} \P[\tau_{\Theta_j,L_j} > a_k]
    <
    \infty,
    $$
    with
    $
    a_k = O\big( k^{c_1} (\log_2 k)^{c_0} \big)
    $
    for some constants $c_0 >0$ and $c_1\ge 1$, then \ref{eq:chebyshev} can be obtained with computational complexity
    $$
O\big( \epsilon^{-2 c_1}\abs{\log_2 \epsilon}^{c_0} \big)
      \qquad
      \text{as}
      \quad
      \epsilon \rightarrow 0.
$$
  \end{proposition}
  \begin{remark}The above result shows that even for costs with unbounded tails, reasonable confidence intervals and complexity order may be possible.  This may be the case for example when a rejection sampler or adaptive resampling mechanism is used within Algorithm \ref{alg:pf} or \ref{alg:is-mlmc}, which may lead to large costs for some $\Theta_k$, for example a cost with a geometric tail.
  \end{remark}
  
  The next results are as in Proposition 4 and 5 of \cite{rhee-glynn} in the standard rMLMC setting, and shows how one can choose $\mathbf{p}$, assuming an additional almost sure cost bound,  so that $\sigma^2<\infty$, with reasonable complexity.  
  
    \begin{proposition}\label{prop:subcanonical-nonoptimised}
      Suppose that the assumptions of Proposition \ref{prop:clt} hold with $\var(P, \mu_{\bar{f}})<\infty$, and that Assumption \ref{a:rates} holds with $\beta\le 1$, where moreover $\tau_{\theta,\ell} \le C 2^{\gamma\ell(1+\rate)}$ almost surely, uniformly in $\Theta_k=\theta\in \mathsf{T}$.  For all $q>2$ and $\eta>1$, the choice of probability
      $$
      p_\ell \propto
      2^{-2 b \ell} \ell [ \log_2 (\ell +1)]^\eta,
      $$
      where $b \defeq \min((\beta + \rate)/2,\alpha)$, leads to $\sigma^2 <\infty$, and \ref{eq:chebyshev} can be obtained with computational complexity
      $$
      O\Big(
      \epsilon^{-\costrate\frac{(1+\rate)}{b}}
      \abs{ \log_2 \epsilon}^{q \costrate \frac{(1+\rate)}{2b}}
      \Big)
      \qquad
      \text{as}
      \quad
      \epsilon \rightarrow 0.
      $$
      \end{proposition}
    \begin{remark} Regarding Proposition \ref{prop:subcanonical-nonoptimised}, with $\costrate=1$:
      \begin{quote}
      \begin{enumerate}[(i)]
\item Under Assumption \ref{a:rates} with $\rate=0$, the usual setup in MLMC before variance reduced estimators are used, the above proposition shows that finite variance and \ref{eq:chebyshev} can be obtained without increasing the number of particles at the higher levels, even in the subcanonical regime.  We have in this case $b=\beta/2\le \alpha$ and complexity
$
      O\Big(
      \epsilon^{-\frac{2}{\beta}}
      \abs{ \log_2 \epsilon }^{ \frac{q}{\beta}}
      \Big).
      $
      When $\beta=1$ (borderline case), dMLMC gives complexity $O(\epsilon^{-2} \abs{\log_2 \epsilon}^2)$ \cite{giles-or,mlpf}, which is negligibly better (recall $q>2$), but is biased inference.
     \item When $\alpha>\beta/2$, which is the usual case in the subcanonical regime ($\beta \le 1$); see \cite{kloeden-platen}, a more efficient use of resources can be obtained by increasing the number of particles (see Proposition \ref{prop:subcanonical-optimised} below).
\end{enumerate}
\end{quote}
      \end{remark}
    \begin{proposition}\label{prop:subcanonical-optimised}
Suppose the assumptions of Proposition \ref{prop:subcanonical-nonoptimised} hold, where moreover $\rate\ge 0$ may vary as a free parameter without changing the constant $C>0$.  Then,
      for all $q>2$, $\eta>1$ constants, the choice $\rate = 2\alpha-\beta$ and probability
      $$
p_\ell \propto 2^{-2\alpha \ell} \ell [\log_2(\ell +1)]^\eta,
      $$
leads to $\sigma^2<\infty$, and \ref{eq:chebyshev} can be obtained with computational complexity
$$
O\Big( \epsilon^{\costrate[-2 - \frac{(1-\beta)}{\alpha} ]} \abs{\log_2 \epsilon}^{\costrate[q + \frac{(1-\beta)}{2\alpha} ]}\Big)
      \qquad
      \text{as}
      \quad
      \epsilon \rightarrow 0.
$$
    \end{proposition}

    \section{Numerical simulations}\label{sec:experiments}

Now the theoretical results relating to the method herein introduced 
will be demonstrated on three examples.  
We will consider one example in the canonical regime, and two in the sub-canonical.  
In the first two experiments, the likelihoods can be computed exactly, 
so that the ground truth $\pi^{(\infty)}(\funfullmain)$ can be easily calculated to arbitrary precision.  We run each example with $100$ independent replications, and calculate the mean squared error (MSE) when the chain is at length $\nriii$ as
$$
\text{MSE}(\nriii)
=
\frac{1}{100} \sum_{i=1}^{100} \big| E_{\nriii,\nrp,\mathbf{p}}^{(i)}(\funfullmain) - \pi^{(\infty)}(\funfullmain)\big|^2,
$$
which is depicted as the thick red line, average of the thin lines, in Figure \ref{fig:ou-gbm} below.
The error decays with the optimal rate of cost$^{-1}$ and log(cost)cost$^{-1}$ in the canonical and sub-canonical cases, respectively, 
where cost is the realised cost of the run, $\cost(\nriii)$ from Section \ref{sec:efficiency}, measured in seconds, with $\nriii$ iterations of the Markov chain.

Three examples will be considered. 
First we consider two models where the exact marginal likelihood can be computed. 
This way a reliable ground truth can be computed with a long MCMC chain, 
providing a strong verification of the theoretical results.
In Section \ref{sec:experiments-ou} the Ornstein--Uhlenbeck (OU) process is considered,
where Euler-Maruyama provides the canonical convergence regime.
In Section \ref{sec:gbm} the Geometric Brownian motion is considered,
where Euler-Maruyama provides sub-canonical convergence regime.
In Section \ref{sec:nonrev} we consider a more complicated 2$d$ model
which does not allow exact computation of the marginal likelihood.
  
It is of interest to compare our methodology to existing unbiased methods. 
 The method we consider for comparison is PMMH using 
 the exact method introduced by Fearnhead et al.~\cite{fearnheadPR}. 
 In their work they provide a way to construct unbiased estimates without approximating the transition density.
 The key idea is to assign to each particle a random positive weight which is an unbiased estimator of the true weight. 
 This method is referred to as the random weight particle filter.
 It has been later extended to continuous-time observations
 in \cite{fearnheadPRS}.
The method of \cite{fearnheadPR} is implemented when it is applicable 
(in particular, for the models of Sections \ref{sec:experiments-ou} and \ref{sec:gbm})
and the corresponding MSE is plotted in comparison to 
our method. 
That method is not amenable to the example of Section \ref{sec:nonrev}.
In all the examples, for the sake of comparison with the standard approach, we also implement a finely-discretised PMMH.  This shows the benefit of our alternative approach based on a coarsely-discretised PMMH with  multilevel IS correction.

\subsection{Ornstein--Uhlenbeck process}\label{sec:experiments-ou}
Consider the OU  process
\begin{equation}\label{eq:ou}
\ud Z_t = -a Z_t \ud t + b \ud W_t \, ,
\qquad t\ge 0,
\end{equation}
with initial condition $Z_0=0$, model parameter $\theta = (\theta_1,\theta_2) \sim N(0,\sigma^2 I)$, and $a\defeq a_\theta=\exp(\theta_1)$ and $b\defeq b_\theta=\exp(\theta_2)$.
The process is discretely observed for $k=1,\dots,\nrt$,
\begin{equation}\label{eq:obs}
Y_p = X_p + \xi_p \, ,
\end{equation}
where $\xi_p \sim N(0,\gamma^2)$ i.i.d.~and recall that $X_p=Z_{p+1}$.  
Therefore, 
$$
G_{p}(x) = \exp(-\frac1{2\gamma^2} |x - y_p|^2 ) \, .
$$

The marginal likelihood is given by 
$$
\bbP[y_{1:\nrt}|\theta] = \prod_{p=1}^\nrt \bbP[y_p|y_{1:p-1},\theta] \, ,
$$
and each factor can be computed as the marginal of the joint on the prediction and current observation, i.e.
\begin{equation}\label{eq:marglike}
\bbP[y_p|y_{1:p-1},\theta] = \int_\bbR \bbP[y_p|x_p,\theta] \bbP[x_p | y_{1:p-1},\theta] \ud x_p \, .
\end{equation}

In this example the ground truth can be computed exactly via the Kalman filter.  
In particular, the solution of \ref{eq:ou} is given by 
$$
Z_1 = e^{-a}X_0 + W_1 \, , \quad W_1 \sim \mathcal{N}\left (0, \frac{b^2}{2a}(1 - e^{-2a}) \right ) \, .
$$
The filter at time $p$ is given by the following simple recursion
$$
m_p = c_p \left (\frac{y_p}{\gamma^2} + \frac{\hat{m}_p}{\hat{c}_p}\right ),~c_p = (\gamma^{-2} + \hat{c}_p^{-1})^{-1} \, , \quad
\hat{m}_p = e^{-a} m_{p-1} \, , \quad \hat{c}_p = e^{-2a} c_{p-1} + \frac{b^2}{2a}(1 - e^{-2a}) \, .
$$ 
Additionally, the incremental marginal likelihoods \ref{eq:marglike} can be computed exactly
$$
\bbP[y_p|y_{1:p-1},\theta] = \sqrt{\frac{c_p}{2\pi\hat{c}_p\gamma^2}}
\exp\left \{-\frac12 \left [ \frac{y_p^2}{\gamma^2} + \frac{\hat{m}_p^2}{\hat{c}_p} - c_p \left ( \frac{y_p}{\gamma^2} + \frac{\hat{m}_p}{\hat{c}_p} \right)^2 \right] \right \} \, .
$$

The parameters are chosen as $\gamma=1$, $\sigma^2=0.1$, $\nrt=5$, and 
the data is generated with $\theta= (0,0)^T$.
Our aim is to compute $\bbE (\theta| y_{1:\nrt})$ (or $\bbE[ (a,b)^T | y_{1:\nrt} ]$, etc., but we will content ourselves with the former).
This is done via a brute force random walk MCMC for $\nriii=10^8$ steps using the exact likelihood $\bbP[y_{1:\nrt}|\theta]$ as above.  
The IACT is around 10, so this gives a healthy limit for MSE computations.

For the numerical experiment, we use Euler-Maruyama method at resolution $h_\ell=2^{-\ell}$
to solve \ref{eq:ou} as follows
\begin{equation}\label{eq:euler}
Z_{p+1} = (1-a h_\ell) Z_p + b B_{p+1} \, , ~~ B_{p+1} \sim \mathcal{N}(0,h_\ell) ~~ i.i.d.
\end{equation}
for $p=1,\dots, K_\ell=h^{-1}_\ell$. Levels $\ell$ and $\ell-1$ are coupled in the simulation of $\Delta_\ell$
by defining $B^C_{1:K_\ell/2} = B^F_{1:2:K_\ell-1} + B^F_{2:2:K_\ell}$
Algorithm \ref{alg:delta-pf} is then run using the standard bootstrap particle filter (Algorithm \ref{alg:pf}) with $\nrp=20$ particles and $O(\nrp)$-complexity multinomial resampling; see \cite{cappe}.  Theorem \ref{theo:mainthm} provides a rate of $\beta=2$ for Algorithm \ref{alg:delta-pf}, because the diffusion coefficient is constant, which implies we are essentially running a Milstein scheme (see \ref{eq:coup_h_cont} and \cite{kloeden-platen}).  Recommendation \ref{rec:allocations} (or Proposition \ref{prop:efficiency-both-finite}) of Section \ref{sec:efficiency} suggests arbitrary precision can be obtained by Algorithm \ref{alg:is-mlmc} with $p_\ell \propto 2^{-3\ell/2}$ and no scaling of particle numbers based on $\ell$ in this canonical $\beta=2$ regime (with weak rate $\alpha=1$).  We choose a positive PMMH algorithm constant $\epsilon=10^{-6}$ (see Remark \ref{alg:is-mlmc-remark}\ref{alg:is-mlmc-remark-pmmh}).  We run Algorithm \ref{alg:is-mlmc} for $10^{4}$ steps, 
with 100 replications. For the finely-discretised PMMH experiment we run $10^4$ steps, with 100 replications with a discretisation of $h_{\ell} = 2^{-5}$.  The results are presented in Figure \ref{fig:ou-gbm}, where it is clear that the theory holds and the MSE decays according to $1/$cost.
The variance of the run-times is very small over replications.
The method of \cite{fearnheadPR}, within PMMH,
also converges with the theoretically-predicted canonical rate, 
but with a slightly smaller constant. 
This is not unexpected. 
The important point is that both methods achieve the same canonical rate, 
while our method is quite generally applicable, 
in particular, to a wide range of models inaccessible to methods of the type of \cite{fearnheadPR}.
Also with the finely discretised PMMH from Figure \ref{fig:ou-gbm}, i.e. the curve titled as MSE3, we see the bias kicking in at the end. We also see that for the same level of cost the MSE is higher than that of the other methodologies.

\subsection{Geometric Brownian motion}\label{sec:gbm}
We next consider the following stochastic differential equation 
\begin{equation}\label{eq:gbm}
\ud Z_{t} = 
 a Z_{t} \ud W_{t}, 
\end{equation}
with initial condition $Z_0 = 1$, and $a\defeq a_\theta=\exp(\theta)$ with $\theta\sim\mathcal{N}(0,\sigma^2)$.  
\begin{figure}
\includegraphics[width=.327\columnwidth]{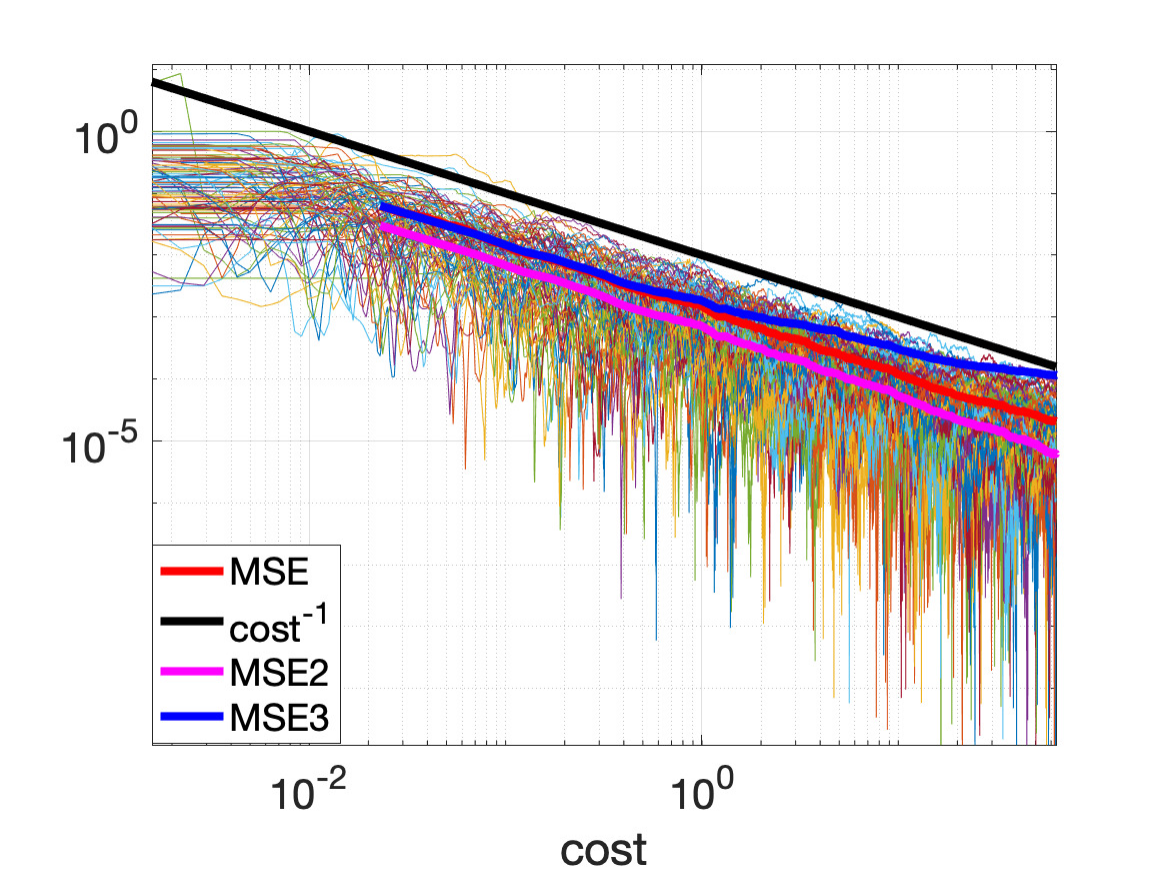}
\includegraphics[width=0.327\columnwidth]{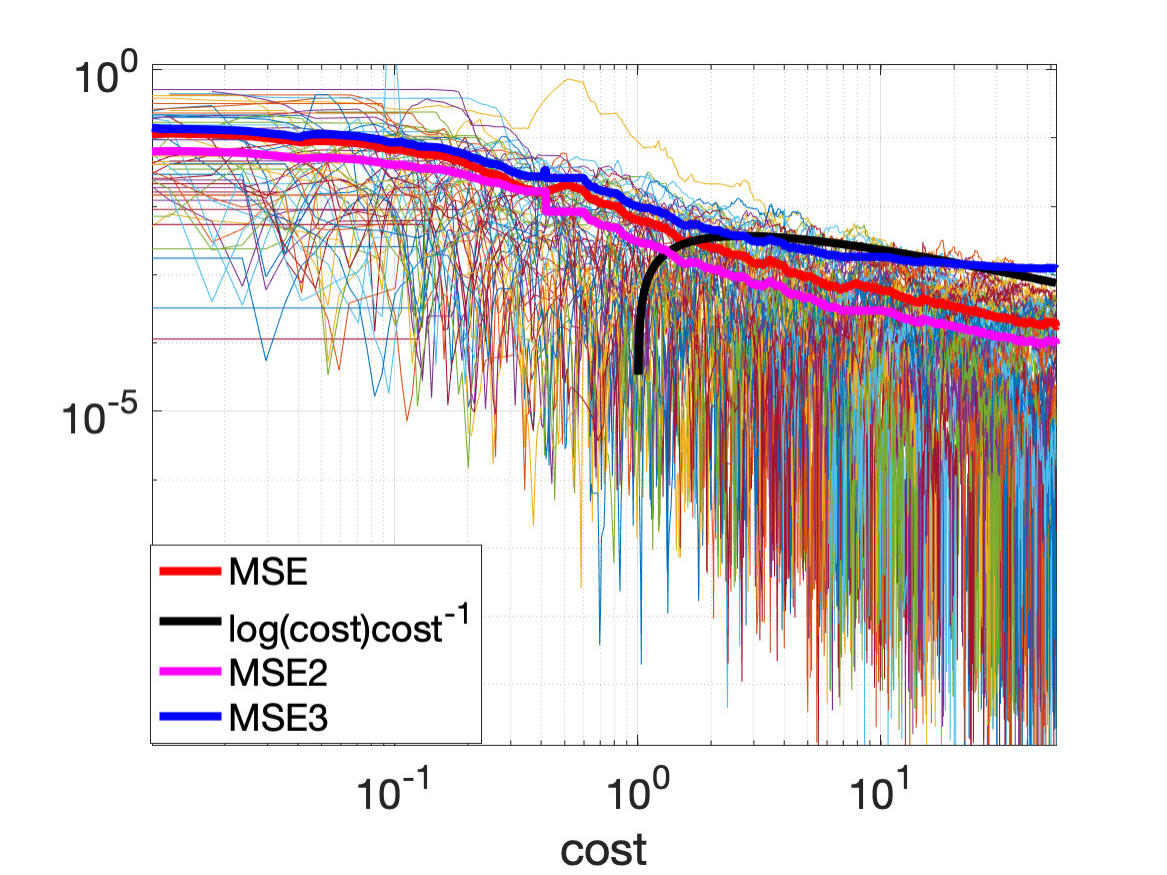}
\includegraphics[width=0.327\columnwidth]{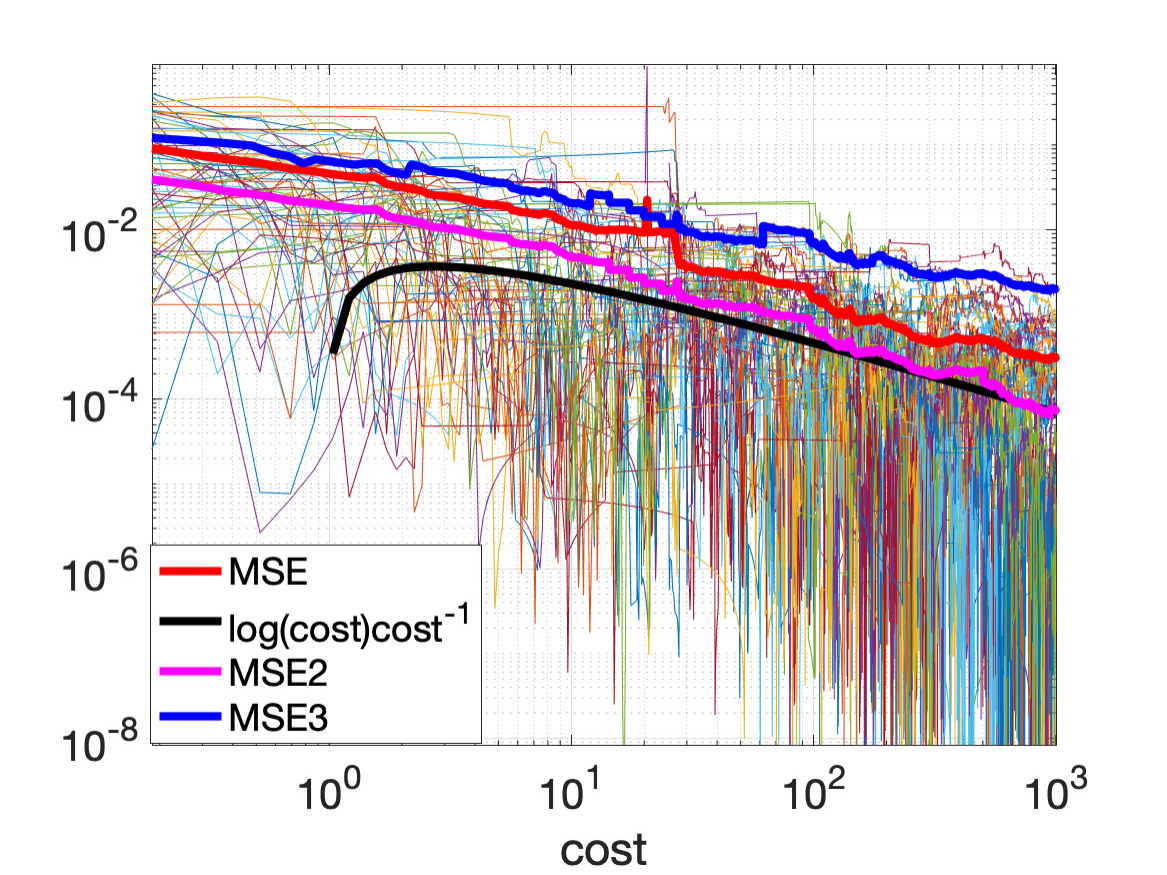}
\caption{The MSE of PMMH rMLMC IS (Algorithm \ref{alg:is-mlmc}) applied to the problem of parameter inference for the discretely observed OU process (left plot) and GBM process (middle plot with $\rho=0$, right plot with $\rho=1$). Squared error replications are given by the thin curves, 
while the thick red curves give the MSE over replications. 
The thick magenta curves show the MSE of
PMMH using 
\cite{fearnheadPR} (denoted by MSE2 in the legend). The blue curve, denoted by MSE3, depicts the MSE for the finely-discretised PMMH.
The black curves representing cost$^{-1}$ (left plot) and log(cost)cost$^{-1}$ (middle and right plots) 
are there to guide the eye.}
\label{fig:ou-gbm}
\end{figure}
This equation is analytically tractable as well, and the solution
of the transformed equation $U = \log Z$ is given via It{\^o}'s formula by
\[
\ud U_{t} = -\frac{a^2}{2} \ud t + a\, \ud W_{t}. 
\]
Defining $W_p \sim \mathcal{N}(0,1)$ i.i.d., one has that 
\[
U_{p+1} = U_p +  -\frac{a^2}{2} + a\, W_p \, , \quad {\rm with} \quad  U_0 = \log z_0 = 0, 
\]
and the solution of~\ref{eq:gbm} can be obtained via exponentiation: $Z_{p} = e^{U_{p}}$.
Moreover, noisy observations 
are introduced on the form , with $X_p=Z_{p+1}$,
\[
Y_p = \log(X_p) + \xi_p,
\]
where $\xi_p \sim \mathcal{N}(0,\gamma^2)$ i.i.d.~as above.
Therefore we have
\begin{equation}\label{eq:g_gbm}
G_{p}(x) = \exp(-\frac1{2\gamma^2} |\log(x) - y_p|^2 ).
\end{equation}
Again $\bbP[y_{1:\nrt}|\theta]$ can be computed analytically. 
The parameters $\gamma=1$, $\sigma^2=0.1$, $\nrt=5$ are chosen 
the same as in the previous example and the true observations are generated again with $\theta=0$.

In order to investigate the theoretical sub-canonical rate, we return to \ref{eq:gbm}
and approximate this directly using Euler-Maruyama method \ref{eq:euler}, which introduces
artificial approximation error.  This problem suffers from stability problems when $X<0$, 
so we take $h_\ell=2^{-6-\ell}$.
 Algorithm \ref{alg:pf} is then used along with the selection functions \eqref{eq:g_gbm}.
Here the diffusion coefficient is not constant, and Euler-Maruyama method yields a rate of $\beta=1=\alpha$,
the borderline case, which is expected to give a logarithmic penalty.
Based on Recommendation \ref{rec:allocations} (or Proposition \ref{prop:subcanonical-optimised}) of Section \ref{sec:efficiency}, 
we consider scaling the particles as $2^{\rate \ell}$ with $\rate=2\alpha-\beta=1$ and $\rate=0$, 
with $p_\ell \propto 2^{-2\ell} \ell \log(\ell)^2$ in both cases.  Again we let $\epsilon=10^{-6}$, and 
 the standard bootstrap particle filter is used, with $\nrp=20\times 2^{\rate \ell}$ particles.  
Algorithm \ref{alg:is-mlmc} is run for $10^{4}$ steps, 
with 100 replications. Again, for the finely-discretised PMMH, we run $10^4$ steps, with 100 replications with a discretisation of $h_{\ell} = 2^{-4}$ 
For this sub-canonical case we impose an artificial upper bound $\ell \leq 10$, 
corresponding to an induced bias of $\approx 10^{-5}$.
The results are presented in Figure \ref{fig:ou-gbm}, and they 
show good agreement with the theory, in terms of rate. 
On the other hand, the cost for $\rate=0$ is apparently smaller than that of 
$\rate=1$ by a factor of approximately 100. 
The method of \cite{fearnheadPR} is not expected to suffer from 
a logarithmic penalty on the MSE convergence, i.e. it achieves canonical rate also in this example.
This can be seen in Figure \ref{fig:ou-gbm}, in addition to a slightly better constant, as before.
For the finely-discretised PMMH from Figure \ref{fig:ou-gbm}, with geometric Brownian motion, we again notice the effect of the bias arising from the discretisation, and the overall higher MSE.

    \subsection{2\textbf{d} Non-reversible  Langevin equation}\label{sec:nonrev}
   
    We now consider a $2d$ example which is not amenable to approaches of the type \cite{fearnheadPR}.
    Consider a target distribution of the type $\rho(z) \propto \exp(-\Phi(z)/a_2)$, $a_2>0$, 
    and the following non-reversible Langevin equation
    \begin{equation}
    \ud Z_t =
    (A - I_2) \nabla \Phi(Z_t) \ud t + \sqrt{2 a_2} 
    \ud W_t,
    \qquad t\ge 0 \, ,
    \end{equation}
and noisy observations $Y_p \sim \mathcal{N}(X_p, \gamma^2 I_2)$, with $\gamma=1, n=10$ observations where $A \in \bbR^{2\times 2}$
is anti-symmetric and parameterised by $a_1 \ge 0$, $I_2$ is the $2$-dimensional identity matrix,
and $\Phi(z) = \frac{a_3}{2}(z_1^2+z_2^2-1)^2$ is the ring potential, parameterised by $a_3 >0$.
The initial condition is specified as $Z_0=[1,1]^T$. 
It is easy to see that the right-hand side of the Fokker-Planck equation vanishes for the invariant distribution
$\rho$ given above, so the dynamics are well-behaved \cite{pavliotis}. We let $a_i=\exp(\theta_i)$, for $i=1,\dots, 3$, and the prior is given by 
$\theta \sim N(0,\varepsilon^2 I_3)$, $\varepsilon^2=0.1$. The resolution of the Euler--Maruyama scheme for this experiment is set as $2^{-\ell+1}$.
This problem is no longer analytically soluble, so a high-resolution simulation is used with a large sample size 
as ground truth.

\begin{figure}
\centering
\includegraphics[width=0.50\textwidth]{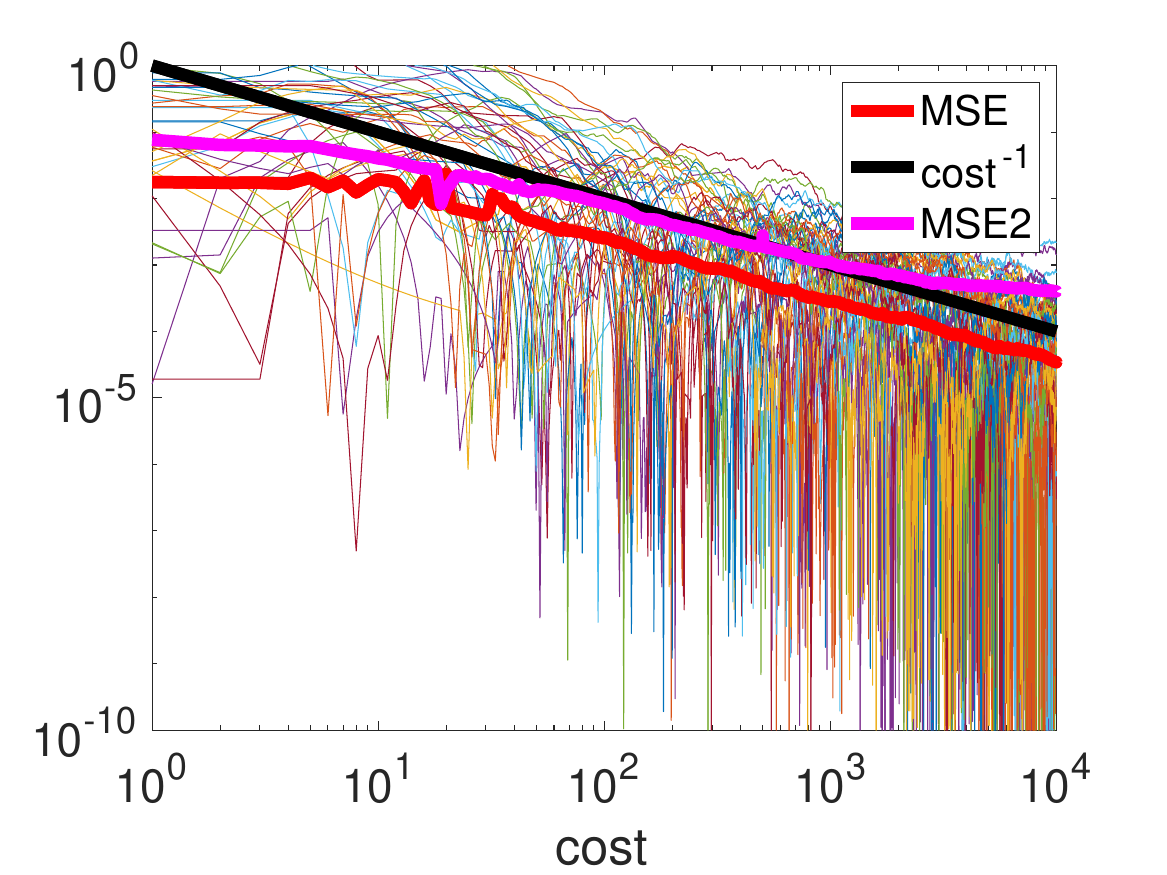}
\caption{The MSE of PMMH rMLMC IS Algorithm \ref{alg:is-mlmc} applied to the problem of parameter inference for the discretely observed non-reversible Langevin equation. Squared error replications are given by the thin curves,  while the thick red curves give the MSE over replications. The thick blue curve shows the MSE over replications of the finely-discretised PMMH, which we denote as MSE2. The black curve represents cost$^{-1}$.}
\label{fig:nonrev}
\end{figure}

For our setup it follows similarly to that of the previous experiment for the OU process, where we are working in a canonical regime. Again we run Algorithm \ref{alg:is-mlmc} with $10^4$ steps, and 100 replications, 
and Algorithm \ref{alg:delta-pf} is run using the standard bootstrap particle filter (Algorithm \ref{alg:pf}) with $\nrp=20$ particles. We compare our results to the single level PMMH, with a similar setup where we specify its resolution as $2^{-4}$. As before we choose a positive PMMH algorithm constant $\epsilon=10^{-6}$. As we can see from Figure \ref{fig:nonrev} the MSE of the proposed methodology in the paper decays at the rate of $1/\textrm{cost}$, which is as expected, which outperforms that of the single-level PMHH.
\begin{remark}
For multidimensional diffusions, it is well known that the exact methodology works only on specific diffusions, which require strong assumptions \cite{beskos-roberts}.  It is not so clear how the exact methodology can be applied to our non-reversible diffusion (due to the difficult drift term) \cite{beskos-papaspiliopoulos-roberts-fearnhead}.  An alternative to this is the work of Blanchet el al. \cite{blanchet-zhang}, which does not rely on the assumption of drift term equal to the gradient of a suitable potential function, or on Lamperti's transformation for that matter. However, despite this, the major drawback is that the running time of their methodology, although finite with probability one, has infinite mean. As a result, the comparison would not be practical due to the cost of the experiment.
\end{remark}

\section*{Acknowledgments} 
JF, AJ, KL and MV have received support from the Academy of Finland (grants 274740, 312605 and 315619) and from the Institute for Mathematical Sciences, Singapore (2018 programme `Bayesian Computation for High-Dimensional Statistical Models'). NC and AJ have received support from KAUST baseline funding, JF and KL from The Alan Turing Institute, AJ from the Singapore Ministry of Education (R-155-000-161-112), and KL from the University of Manchester (School of Mathematics). This research made use of the Rocket High Performance Computing service at Newcastle University. We thank Paul Fearnhead, Santeri Karppinen, Anthony Lee and the anonymous referees for their many insightful remarks.

\appendix

\section{Analysis of the delta particle filter}\label{app:bound}

We now give our analysis that is required for the proofs of Theorem \ref{theo:mainthm} and Corollary \ref{theo:maincor} of Section \ref{sec:bound} regarding the $\Delta$PF (Algorithm \ref{alg:delta-pf}) for HMM diffusions.  
The structure of the appendix is as follows. 
In Section \ref{app_sec:models} we introduce some more Feynman--Kac notations, following \cite{del-moral,mlpf}, emphasising that here we consider standard HMMs that can be coupled.
In Section \ref{sec:delta_pf} we recall the $\Delta$PF stated earlier.  A general variance bound for quantities such as $\Delta_\ell(\varphi)$ is given in Section
\ref{app_sec:delta_gen}.  This is particularised to the HMM diffusion case in Section \ref{sec:diff_case}, where we supply the proofs for the results of Section \ref{sec:bound}.

\subsection{Models}\label{app_sec:models}

Let $(\mathsf{X},\mathcal{X})$ be a measurable space and $\{G_n\}_{n\geq 0}$ a sequence of non-negative, bounded and measurable functions such that $G_n:\mathsf{X}\rightarrow\mathbb{R}_+$.
Let $\eta_0^{F},\eta_0^C\in\mathscr{P}(\mathsf{X})$ and $\{M_n^F\}_{n\geq 1}$,   $\{M_n^C\}_{n\geq 1}$ be two sequences of Markov kernels, i.e.~$M_n^F:\mathsf{X}\rightarrow\mathscr{P}(\mathsf{X})$,
$M_n^C:\mathsf{X}\rightarrow\mathscr{P}(\mathsf{X})$. 
Set $\mathsf{E}_n \defeq\mathsf{X}^{n+1}$ for $n\geq 0$, and for $x_{0:n}\in\mathsf{E}_n$,
$$
\pmb{G}_n(x_{0:n}) = G_n(x_n)
$$
and for $n\geq 1$, $s\in\{F,C\}$, $x_{0:n-1}\in\mathsf{E}_{n-1}$
$$
\pmb{M}_n^{s}(x_{0:n-1},\ud x_{0:n}') = \delta_{\{x_{0:n-1}\}}(\ud x_{0:n-1}') M_n^s(x_{n-1}',\ud x_n').
$$
Define for $s\in\{F,C\}$, $\varphi\in\mathcal{B}_b(\mathsf{E}_n)$, $u_n\in\mathsf{E}_n$
$$
\pmb{\gamma}_{n}^s(\varphi) = \int_{\mathsf{E}_0\times\cdots\times\mathsf{E}_n}\varphi(u_n) \Big(\prod_{p=0}^{n-1} \pmb{G}_p^s(u_p)\Big) \eta_0^s(\ud u_0)\prod_{p=1}^n \pmb{M}_p^s(u_{p-1},\ud u_p)
$$
and
$$
\pmb{\eta}_n^s(\varphi) = \frac{\pmb{\gamma}_{n}^s(\varphi)}{\pmb{\gamma}_{n}^s(1)}.
$$

Throughout this appendix, we assume Assumption \ref{a:diffusion}, and that
Assumption \ref{a:coupled-fk}\eqref{a:coupled-fk-kernel} holds, i.e.~there exists $\check{\eta}_0\in\mathscr{P}(\mathsf{X}\times\mathsf{X})$ such that
for any $A\in\mathcal{X}$
$$
\check{\eta}_0(A\times \mathsf{X}) = \eta_0^F(A) \qquad \check{\eta}_0(\mathsf{X}\times A) = \eta_0^C(A)
$$
and moreover for any $n\geq 1$ there exists Markov kernels $\{\check{M}_n\}$, $\check{M}_n:\mathsf{X}\times\mathsf{X}\rightarrow\mathscr{P}(\mathsf{X}\times\mathsf{X})$
such that for any $A\in\mathcal{X}$, $(x,x')\in\mathsf{X}\times\mathsf{X}$:
\begin{equation}\label{eq:coupling-mss}
\check{M}_n(A\times \mathsf{X})(x,x') = M_n^F(A)(x) \qquad \check{M}_n(\mathsf{X}\times A)(x,x') = M_n^C(A)(x').
\end{equation}

\subsection{Delta particle filter}\label{sec:delta_pf}
Define $x_p=(x_p^F,x_p^C)\in\mathsf{X}\times\mathsf{X}$ and
$$
\check{G}_p(x_p) = \frac{1}{2}(G_p(x_p^F)+G_p(x_p^C)),
$$
as in Assumption \ref{a:coupled-fk}\eqref{a:coupled-fk-potential}.
Set, for $n\geq 0$, $x_{0:n}\in\mathsf{X}^{2(n+1)}$
$$
\pmb{\check{G}}_n(x_{0:n}) = \check{G}_n(x_n)
$$
and for $n\geq 1$, $x_{0:n-1}\in\mathsf{X}^{2n}$
$$
\pmb{\check{M}}_n(x_{0:n-1},\ud x_{0:n}') = \delta_{\{x_{0:n-1}\}}(\ud x_{0:n-1}') \check{M}_n(x_{n-1}',\ud x_n'),.
$$
Note that coupling assumption \eqref{eq:coupling-mss} for $\check{M}_n$ can be equivalently formulated for $\pmb{\check{M} }_n$.

For $n\geq 0$, $\varphi\in\mathcal{B}_b(\mathsf{E}_n\times\mathsf{E}_n)$, $u_n\in\mathsf{E}_n\times\mathsf{E}_n$, we have
$$
\pmb{\check{\gamma}}_{n}(\varphi) = \int_{\mathsf{E}_0^2\times\cdots\times\mathsf{E}_n^2}\varphi(u_n) \Big(\prod_{p=0}^{n-1} \pmb{\check{G}}_p(u_p)\Big) \check{\eta}_0(\ud u_0)\prod_{p=1}^n \pmb{\check{M}}_p(u_{p-1},\ud u_p)
$$
and
$$
\pmb{\check{\eta}}_n(\varphi) = \frac{\pmb{\check{\gamma}}_{n}(\varphi)}{\pmb{\check{\gamma}}_{n}(1)}.
$$

As noted in \cite{jasra-kamatani-law-zhou} it is simple to establish that for $\varphi\in\mathcal{B}_b(\mathsf{E}_{n})$, if
\begin{equation}\label{eq:psi_def}
\psi(x_{0:n}) =
\pmb{\check{G}}_n(x_{0:n})\Big(\varphi(x_{0:n}^F)\prod_{p=0}^n \frac{\pmb{G}_p(x_{0:p}^F)}{\pmb{\check{G}}_p(x_{0:p})}
-
\varphi(x_{0:n}^C)\prod_{p=0}^n \frac{\pmb{G}_p(x_{0:p}^C)}{\pmb{\check{G}}_p(x_{0:p})}
\Big)
\end{equation}
then
\begin{equation}\label{eq:main_id}
\pmb{\check{\gamma}}_{n}(\psi) = \pmb{\check{\gamma}}_{n}(1)\pmb{\check{\eta}}_n(\psi)
=
\gamma_n^F(G_n\varphi) - \gamma_n^C(G_n\varphi).
\end{equation}
Note 
$$
 \pmb{\check{\gamma}}_{n}(1) = \prod_{p=0}^{n-1}\pmb{\check{\eta}}_p(\pmb{\check{G}}_p).
$$

 In order to approximate $\pmb{\check{\gamma}}_{n}(\psi)$ one can run the following abstract version of Algorithm \ref{alg:delta-pf} (recall from Section \ref{sec:bound} that we will only consider multinomial resampling). Define for $n\geq 1$, $\mu\in\mathscr{P}(\mathsf{E}_{n-1}\times\mathsf{E}_{n-1})$, $\varphi\in\mathcal{B}_b(\mathsf{E}_n\times\mathsf{E}_n)$
$$
\pmb{\check{\phi}}_n(\mu)(\varphi) = \frac{\mu(\pmb{\check{G}}_{n-1}\pmb{\check{M}}_n(\varphi))}{\mu(\pmb{\check{G}}_{n-1})}.
$$
The algorithm begins by generating $u_0^i\in\mathsf{E}_0\times\mathsf{E}_0$, $i\in\{1,\dots,N\}$ with joint law
$$
\prod_{i=1}^N \check{\eta}_0(\ud u_0^i) = \prod_{i=1}^N \pmb{\check{\eta}}_0(\ud u_0^i).
$$
Defining 
$$
\pmb{\check{\eta}}_0^N(\ud u_0) = \frac{1}{N}\sum_{i=1}^N \delta_{u_0^i}(\ud u_0)
$$
we then generate  $u_1^i\in\mathsf{E}_1\times\mathsf{E}_1$, $i\in\{1,\dots,N\}$ with joint law
$$
\prod_{i=1}^N \pmb{\check{\phi}}_1(\pmb{\check{\eta}}_0^N)(\ud u_1^i).
$$
This proceeds recursively, so the joint law of the particles up to time $n$ is
$$
\Big(\prod_{i=1}^N \pmb{\check{\eta}}_0(\ud u_0^i)\Big)\Big(\prod_{p=1}^n\prod_{i=1}^N \pmb{\check{\phi}}_p(\pmb{\check{\eta}}_{p-1}^N)(\ud u_p^i)\Big).
$$
Hence we have the estimate
$$
\pmb{\check{\gamma}}^N_{n}(\psi) = \Big(\prod_{p=0}^{n-1}\pmb{\check{\eta}}_p^N(\pmb{\check{G}}_p)\Big)\pmb{\check{\eta}}^N_n(\psi).
$$
\begin{remark}\label{rem:key_point}
Note that $\pmb{\check{\gamma}}^N_{n}(\psi)$ corresponds to the quantity $\Delta_\ell(\varphi)$ in \eqref{eq:Delta} from the $\Delta$PF output (Algorithm \ref{alg:delta-pf}).
\end{remark}

\subsection{General hidden Markov model case}\label{app_sec:delta_gen}

Define for $p\geq 1$ the semigroup
$$
\pmb{\check{Q}}_p(x_{0:p-1},\ud x_{0:p}') = \pmb{\check{G}}_{p-1}(x_{0:p-1}) \pmb{\check{M}}_p(x_{0:p-1},\ud x_{0:p}')
$$
with the definition for $0\leq p \leq n$, $\varphi\in\mathcal{B}_b(\mathsf{E}_n\times\mathsf{E}_n)$
$$
\pmb{\check{Q}}_{p,n}(\varphi)(u_p) = \int \varphi(u_n) \prod_{j=p+1}^n  \pmb{\check{Q}}_j(u_{j-1},\ud u_{j})
$$
if $p=n$ clearly $\pmb{\check{Q}}_{n,n}$ is the identity operator. For any $0\leq n$, $\varphi\in\mathcal{B}_b(\mathsf{E}_n\times\mathsf{E}_n)$
we set $\pmb{\check{Q}}_{-1,n}(\varphi)(u_{-1})=0$.

Now following \citep[Chapter 7]{del-moral} we have the following martingale (w.r.t.~the natural filtration of the particle system), $\varphi\in\mathcal{B}_b(\mathsf{E}_n\times\mathsf{E}_n)$:
\begin{equation}\label{eq:martingale}
\pmb{\check{\gamma}}^N_{n}(\varphi) - \pmb{\check{\gamma}}_{n}(\varphi) = 
\sum_{p=0}^n  \pmb{\check{\gamma}}_{p}^N(1)[\pmb{\check{\eta}}_p^N-\pmb{\check{\phi}}_p(\pmb{\check{\eta}}_{p-1}^N)](\pmb{\check{Q}}_{p,n}(\varphi))
\end{equation}
with the convention that $\pmb{\check{\phi}}_p(\pmb{\check{\eta}}_{p-1}^N)=\pmb{\check{\eta}}_0$ if $p=0$. The representation immediately establishes
that 
$$
\mathbb{E}[\pmb{\check{\gamma}}^N_{n}(\varphi)]=  \pmb{\check{\gamma}}_{n}(\varphi)
$$
where the expectation is w.r.t.~the law associated to the particle system. 
We will use the following convention that $C'$ is a finite positive constant that does not depend upon $n,N$ or any of the $G_n$, $M_n^{s}$ ($s\in\{F,C\}$, $\check{M}_n$.
The value of $C'$ may change from line-to-line.
Define for $0\leq p \leq n<\infty$
$$
\overline{G}_{p,n} = \prod_{q=p}^{n}\|G_q\|
$$
with the convention that if $p=0$ we write $\overline{G}_{n}$.
We have the following result.

\begin{proposition}\label{prop:gen_case}
  Suppose that $\|G_n\|<\infty$ for each $n\geq 0$. Then there exist a $C'<\infty$ such that for any $n\geq 0$, $\varphi\in\mathcal{B}_b(\mathsf{E}_n\times\mathsf{E}_n)$  
$$
\mathbb{E}\Big[\Big(\pmb{\check{\gamma}}^N_{n}(\varphi) - \pmb{\check{\gamma}}_{n}(\varphi)\Big)^2\Big]
\leq \frac{C'}{N}\sum_{p=0}^n \overline{G}_{p-1}^2 \mathbb{E}[\pmb{\check{Q}}_{p,n}(\varphi)(u_p^1)^2].
$$
\end{proposition}

\begin{proof}
Set 
$$
\pmb{\check{S}}_{p,n}^N(\varphi)
=
\pmb{\check{\gamma}}_{p}^N(1)[\pmb{\check{\eta}}_p^N-\pmb{\check{\phi}}_p(\pmb{\check{\eta}}_{p-1}^N)](\pmb{\check{Q}}_{p,n}(\varphi))
$$
By \eqref{eq:martingale}, one can apply the Burkholder-Gundy-Davis inequality to obtain
\begin{equation}\label{eq:prf1}
\mathbb{E}\Big[\Big(\pmb{\check{\gamma}}^N_{n}(\varphi) - \pmb{\check{\gamma}}_{n}(\varphi)\Big)^2\Big] \leq 
C'
\sum_{p=0}^n 
\mathbb{E}[\pmb{\check{S}}_{p,n}^N(\varphi)^2].
\end{equation}
Now, we have that
$$
\mathbb{E}[\pmb{\check{S}}_{p,n}^N(\varphi)^2]
\leq \overline{G}_{p-1}^2\mathbb{E}[[\pmb{\check{\eta}}_p^N-\pmb{\check{\phi}}_p(\pmb{\check{\eta}}_{p-1}^N)](\pmb{\check{Q}}_{p,n}(\varphi))^2].
$$
Application of the (conditional) Marcinkiewicz-Zygmund inequality yields
$$
\mathbb{E}[\pmb{\check{S}}_{p,n}^N(\varphi)^2]
\leq \frac{C' \overline{G}_{p-1}^2}{N}
\mathbb{E}\Big[
\Big(
  \pmb{\check{Q}}_{p,n}(\varphi)(u_p^1) -
\pmb{\check{\phi}}_p(\pmb{\check{\eta}}_{p-1}^N)( \pmb{\check{Q}}_{p,n}(\varphi) )
  \Big)^2\Big].
$$
After applying $C_2$ and Jensen inequalities, we then conclude by \eqref{eq:prf1}.
\end{proof}

\subsection{Diffusion case}\label{sec:diff_case}

We now consider the model of Section \ref{sec:bound}, where we recall that $\theta$ is omitted from the notation. 
A series of technical results are given and the proofs for Theorem \ref{theo:mainthm} and Corollary \ref{theo:maincor} are given at the end of this section.

We recall that the joint probability density of the observations and the unobserved diffusion at the observation times 
is given by
$$
\prod_{p=0}^n G_{p}(x_{p}) Q^{(\infty)}(x_{p-1},x_{p}).
$$
As the true dynamics can not be simulated, in practice we work with
$$
\prod_{p=0}^n G_{p}(x_p) Q^{(\ell)}(x_{p-1},x_{p}).
$$
Recall an (Euler) approximation scheme
with discretisation $h_{\ell}=2^{-\ell}$, $\ell\geq 0$.
In our context then, $M_n^F$ corresponds $Q^{(\ell)}$ ($\ell\geq 1$) and
$M_n^C$ corresponds $Q^{(\ell-1)}$. 
The initial distribution $\eta_0$ is simply the (Euler) kernel started at some given $x_0$.
As noted earlier in Remark \ref{delta-pf-remark}\eqref{delta-pf-remark-kernel}, a natural coupling of $M_n^F$ and $M_n^C$ (and hence of $\eta_0$) exists. 
As established in \cite[eq.~(32)]{mlpf} one has (uniformly in $\thetaone $ as Assumption \ref{a:diffusion} holds with $\thetaone $ independent constants)
for $C'<\infty$
\begin{equation}\label{eq:markov_cont}
\sup_{\mathcal{A}}\sup_{x\in\mathsf{X}}|M_n^F(\varphi)(x)-M_n^C(\varphi)(x)| \leq C' h_{\ell}
\end{equation}
where $\mathcal{A}=\{\varphi\in\mathcal{B}_b(\mathsf{X})\cap\textrm{Lip}(\mathsf{X}): \|\varphi\|\leq 1|\}$.
We also recall that \eqref{eq:coup_h_cont} holds (recall Assumption \ref{a:diffusion} is assumed).

We will use $M<\infty$ to denote a constant that may change
from line-to-line. It will not depend upon $\theta$ nor $N$, $\ell$, but may depend on the time parameter or a function.
The following result will be needed later on. The proof is given after the proof of Lemma \ref{lem:res1} below.

\begin{proposition}\label{prop:h_smooth}
Assume (A\ref{hyp:1} (i)-(ii),\ref{hyp:2}). Then for any $n\geq 0$ and $\varphi\in\mathcal{B}_b(\mathsf{X}^{n+1})\cap\textrm{\emph{Lip}}(\mathsf{X}^{n+1})$
there exists a $M<\infty$ such that
$$
|\gamma_n^F(G_n\varphi) - \gamma_n^C(G_n\varphi)| \leq M h_{\ell}
$$
\end{proposition}

We write expectations w.r.t.~the time-inhomogeneous Markov chain associated to the sequence of kernels $(M_p^F)_{p\geq 1}$ (resp.~$(M_p^C)_{p\geq 1}$) 
as $\mathbb{E}^F$, (resp.~$\mathbb{E}^C$). 

\begin{lemma}\label{lem:res1}
Assume (A\ref{hyp:1}(i)-(ii),\ref{hyp:2}). Let $s\in\{F,C\}$ and $\varphi\in\mathcal{B}_b(\mathsf{X}^{n+1})\cap\textrm{\emph{Lip}}(\mathsf{X}^{n+1})$,
then, define the function for $0\leq p \leq n$
$$
\varphi_{p,n}^s(x_{0:p}) := \mathbb{E}^s[\varphi(x_{0:p},X_{p+1:n})\prod_{q=p+1}^n G_q(X_q)|x_p].
$$
Then we have that $\varphi_{p,n}^s\in\mathcal{B}_b(\mathsf{X}^{p+1})\cap\textrm{\emph{Lip}}(\mathsf{X}^{p+1})$. 
\end{lemma}

\begin{proof}
The case $p=n$ follows immediately from $\varphi\in\mathcal{B}_b(\mathsf{X}^{n+1})\cap\textrm{Lip}(\mathsf{X}^{n+1})$.
We will use a backward inductive argument on $p$. Suppose $p=n-1$ then we have for any $(x_{0:n-1},x_{0:n-1}')\in \mathsf{X}^n\times \mathsf{X}^n$
$$
|\varphi_{n-1,n}^s(x_{0:n-1})-\varphi_{n-1,n}^s(x_{0:n-1}')| = 
$$
$$
|\mathbb{E}^s[\varphi(x_{0:n-1},X_{n})G_n(X_n)|x_{n-1}]-\mathbb{E}^s[\varphi(x_{0:n-1}',X_{n})G_n(X_n)|x_{n-1}']| \leq
$$
$$
|\mathbb{E}^s[\varphi(x_{0:n-1},X_{n})G_n(X_n)|x_{n-1}]-\mathbb{E}^s[\varphi(x_{0:n-1}',X_{n})G_n(X_n)|x_{n-1}]| + 
$$
$$
|\mathbb{E}^s[\varphi(x_{0:n-1}',X_{n})G_n(X_n)|x_{n-1}]-\mathbb{E}^s[\varphi(x_{0:n-1}',X_{n})G_n(X_n)|x_{n-1}']|
$$
By $\varphi\in\textrm{Lip}(\mathsf{X}^{n+1})$ it easily follows via (A\ref{hyp:1}(i)) that
$$
|\mathbb{E}^s[\varphi(x_{0:n-1},X_{n})G_n(X_n)|x_{n-1}]-\mathbb{E}^s[\varphi(x_{0:n-1}',X_{n})G_n(X_n)|x_{n-1}]| \leq M
\sum_{j=0}^{n-1} |x_j-x_j'|.
$$
By (A\ref{hyp:1}(ii)) and $\varphi\in\textrm{Lip}(\mathsf{X}^{n+1})$, $\varphi(x_{0:n})G_n(x_n)$ is Lipschitz in $x_n$ and hence by (A\ref{hyp:2})
\begin{equation}\label{eq:t11}
|\mathbb{E}^s[\varphi(x_{0:n-1}',X_{n})G_n(X_n)|x_{n-1}]-\mathbb{E}^s[\varphi(x_{0:n-1}',X_{n})G_n(X_n)|x_{n-1}']|
\leq M|x_{n-1}-x_{n-1}'|.
\end{equation}
Hence it follows
$$
|\varphi_{n-1,n}^s(x_{0:n-1})-\varphi_{n-1,n}^s(x_{0:n-1}')| \leq M \sum_{j=0}^{n-1} |x_j-x_j'|.
$$
The induction step follows by almost the same argument as above and is hence omitted.
\end{proof}

\begin{proof}[Proof of Proposition \ref{prop:h_smooth}]
We have the following standard collapsing sum representation:
\begin{eqnarray*}
\gamma_n^F(G_n\varphi) - \gamma_n^C(G_n\varphi)  & = &
\sum_{p=0}^n\Bigg(
\mathbb{E}^F[\prod_{q=0}^p G_q(X_q)\mathbb{E}^C[\varphi(X_{0:n})\prod_{q=p+1}^n G_q(X_q)|X_p]]
 - \\ & & 
\mathbb{E}^F[\prod_{q=0}^{p-1} G_q(X_q)\mathbb{E}^C[\varphi(X_{0:n})\prod_{q=p}^n G_q(X_q)|X_{p-1}]]\Bigg)
\end{eqnarray*}

The summand is
$$
T_p:=\mathbb{E}^F\Big[\Big(\prod_{q=0}^{p-1} G_q(X_q)\Big) (\mathbb{E}^F-\mathbb{E}^C)\Big(\mathbb{E}^C[\varphi(X_{0:n})\prod_{q=p+1}^n G_q(X_q)|X_p]G_p(X_p)\Big|X_{p-1}\Big)\Big].
$$
By Lemma \ref{lem:res1}, $\mathbb{E}^C[\varphi(x_{0:p},X_{p+1:n})\prod_{q=p+1}^n G_q(X_q)|x_p]\in\mathcal{B}_b(\mathsf{X}^{p+1})\cap\textrm{Lip}(\mathsf{X}^{p+1})$
and by (A\ref{hyp:1}) (i) and (ii) $G_p\in\mathcal{B}_b(\mathsf{X})\cap\textrm{Lip}(\mathsf{X})$. So by \eqref{eq:markov_cont}
$$
\Big|(\mathbb{E}^F-\mathbb{E}^C)\Big(\mathbb{E}^C[\varphi(X_{0:n})\prod_{q=p+1}^n G_q(X_q)|X_p]G_p(X_p)\Big|X_{p-1}\Big)\Big|
\leq 
$$
$$
M h_{\ell} \sup_{x_{0:p}\in\mathsf{X}^{p+1}}|\mathbb{E}^C[\varphi(x_{0:p},X_{p+1:n})\prod_{q=p+1}^n G_q(X_q)|
$$
and hence
$$
|T_p| \leq M h_{\ell} \mathbb{E}^F[\prod_{q=0}^{p-1} G_q(X_q)] \sup_{x_{0:p}\in\mathsf{X}^{p+1}}|\mathbb{E}^C[\varphi(x_{0:p},X_{p+1:n})\prod_{q=p+1}^n G_q(X_q)|x_p]G_p(x_p)|.
$$
Application of (A\ref{hyp:1}) (i) gives $|T_p| \leq M h_{\ell}$ and the proof is hence concluded. 
\end{proof}

\begin{lemma}\label{lem:res4}
Assume (A\ref{hyp:1}). Then for any $n\geq 0$ there exists a $M<\infty$ such that for any $x_{0:n}\in\mathsf{X}^{2(n+1)}$
$$
\Big|\prod_{p=0}^n \frac{G_p(x_p^F)}{\check{G}_p(x_p)} - \prod_{p=0}^n \frac{G_p(x_p^C)}{\check{G}_p(x_p)} \Big| \leq M \sum_{p=0}^n |x_p^F-x_p^C|.
$$
\end{lemma}

\begin{proof}
The is proof by induction. The case $n=0$:
$$
\Big|\frac{G_0(x_0^F)}{\check{G}_0(x_0)} - \frac{G_0(x_0^C)}{\check{G}_0(x_0)} \Big| = \frac{1}{\check{G}_0(x_0)}|G_0(x_0^F) - G_0(x_0^C)|.
$$
Application of (A\ref{hyp:1}) (ii) and (iii) yield that
$$
\Big|\frac{G_0(x_0^F)}{\check{G}_0(x_0)} - \frac{G_0(x_0^C)}{\check{G}_0(x_0)} \Big| \leq M |x_0^F-x_0^C|.
$$
The result is assumed to hold at rank $n-1$, then 
$$
\Big|\prod_{p=0}^n \frac{G_p(x_p^F)}{\check{G}_p(x_p)} - \prod_{p=0}^n \frac{G_p(x_p^C)}{\check{G}_p(x_p)} \Big| \le
$$
$$
\Big|\frac{G_n(x_n^F)}{\check{G}_n(x_n)} - \frac{G_n(x_n^C)}{\check{G}_n(x_n)}\Big|\cdot
\prod_{p=0}^{n-1} \frac{G_p(x_p^F)}{\check{G}_p(x_p)}
+ \Big|\prod_{p=0}^{n-1} \frac{G_p(x_p^F)}{\check{G}_p(x_p)} - \prod_{p=0}^{n-1} \frac{G_p(x_p^C)}{\check{G}_p(x_p)}\Big|
\cdot \frac{G_n(x_n^C)}{\check{G}_n(x_n)}.
$$
For the first term of the R.H.S.~one can follow the argument at the initialisation and apply (A\ref{hyp:1}) (i) and (iii). For the second term of the R.H.S., the induction
hypothesis and (A\ref{hyp:1}) (i) and (iii) can be used. That is one can deduce that 
$$
\Big|\prod_{p=0}^n \frac{G_p(x_p^F)}{\check{G}_p(x_p)} - \prod_{p=0}^n \frac{G_p(x_p^C)}{\check{G}_p(x_p)} \Big| \leq M \sum_{p=0}^n |x_p^F-x_p^C|.
$$
\end{proof}

Recall \eqref{eq:psi_def} for the definition of $\psi$ and that
$x_p=(x_p^F,x_p^C)\in\mathsf{X}\times\mathsf{X}$.

\begin{lemma}\label{lem:res2}
Assume (A\ref{hyp:1}-\ref{hyp:2}). Then 
for any $0\leq p < n$, $\varphi\in\mathcal{B}_b(\mathsf{X}^{n+1})\cap\textrm{\emph{Lip}}(\mathsf{X}^{n+1})$ there exists a $M<\infty$ such that for any
$x_{0:p}\in\mathsf{E}_p\times\mathsf{E}_p$
$$
|\pmb{\check{Q}}_{p,n}(\psi)(x_{0:p})| \leq  M \Big(\sum_{j=0}^p |x_j^F-x_j^C| + h_{\ell}\Big)
$$
\end{lemma}

\begin{proof}
We have 
\begin{align*}
\pmb{\check{Q}}_{p,n}(\psi)(x_{0:p}) = \check{G}_p(x_p) \times 
  \Big(&\prod_{q=0}^p \frac{G_q(x_q^F)}{\check{G}_q(x_q)} \mathbb{E}^F[\varphi(x_{0:p}^F,Y_{p+1:n}) \prod_{s=p+1}^n G_s(X_s^F)|x_p^F] \\
    &-\prod_{q=0}^p \frac{G_q(x_q^C)}{\check{G}_q(x_q)} \mathbb{E}^C[\varphi(x_{0:p}^C,Y_{p+1:n})\prod_{s=p+1}^n G_s(X_s^C)|x_p^C]\Big).
\end{align*}
It then follows that $\pmb{\check{Q}}_{p,n}(\psi)(x_{0:p}) = \check{G}_p(x_p) ( T_1+T_2)$ where
\begin{eqnarray*}
T_1 & = & \Big(\prod_{q=0}^p \frac{G_q(x_q^F)}{\check{G}_q(x_q)} - \prod_{q=0}^p \frac{G_q(x_q^C)}{\check{G}_q(x_q)}\Big) \mathbb{E}^F[\varphi(x_{0:p}^F,Y_{p+1:n})\prod_{s=p+1}^n G_s(X_s^F)|x_p^F] \\
T_2 & = & \prod_{q=0}^p \frac{G_q(x_q^C)}{\check{G}_q(x_q)}\Big(\mathbb{E}^F[\varphi(x_{0:p}^F,Y_{p+1:n})\prod_{s=p+1}^n G_s(X_s^F)|x_p^F]-\mathbb{E}^C[\varphi(x_{0:p}^C,Y_{p+1:n})\prod_{s=p+1}^n G_s(X_s^C)|x_p^C]\Big).
\end{eqnarray*}
By Lemma \ref{lem:res4}, $\varphi\in\mathcal{B}_b(\mathsf{X}^{n+1})\cap\textrm{Lip}(\mathsf{X}^{n+1})$ and (A\ref{hyp:1}) (i)
$$
|T_1| \leq M \sum_{j=0}^p |x_j^F-x_j^C|.
$$

Now $T_2 = T_3 + T_4$ where
\begin{eqnarray*}
T_3 & = & \prod_{q=0}^p \frac{G_q(x_q^C)}{\check{G}_q(x_q)}\Big(\mathbb{E}^F[\varphi(x_{0:p}^F,Y_{p+1:n})\prod_{q=p+1}^n G_s(X_s^F)|x_p^F] - 
\mathbb{E}^F[\varphi(x_{0:p}^F,Y_{p+1:n})\prod_{s=p+1}^n G_s(X_s^F)|x_p^C]\Big) \\
T_4 & = &
\prod_{q=0}^p \frac{G_q(x_q^C)}{\check{G}_q(x_q)}\Big(\mathbb{E}^F[\varphi(x_{0:p}^F,Y_{p+1:n})\prod_{s=p+1}^n G_s(X_s^F)|x_p^C]-\mathbb{E}^C[\varphi(x_{0:p}^C,Y_{p+1:n})\prod_{s=p+1}^n G_s(X_s^C)|x_p^C]\Big).
\end{eqnarray*}
For $T_3$ one can use Lemma \ref{lem:res1} (along with (A\ref{hyp:1}) (i) and (iii)) to get that
$$
|T_3| \leq M \sum_{j=0}^p |x_j^F-x_j^C|.
$$
For $T_4$ a similar collapsing sum argument that is used in the proof of Proposition \ref{prop:h_smooth} can be used to deduce that
$$
|T_4| \leq M h_{\ell}.
$$
One can then conclude the proof via the above bounds (along with (A\ref{hyp:1}) (i)).
\end{proof}

Below $\mathbb{E}$ denotes expectation w.r.t.~the particle system described in Section \ref{sec:delta_pf} started at position $(x,x)$ at time $n=0$ with $x\in \mathsf{X}$, in the diffusion case of Section \ref{sec:diff_case}.  Recall the particle $U_n^i\in\mathsf{E}_n\times\mathsf{E}_n$ at time $n\ge 0$ in path space. We denote by $U_n^{i,s}(j)\in\mathsf{X}$ as the $j\in\{0,\dots,n\}$ component of
particle $i\in\{1,\dots,N\}$ at time $n\geq 0$ of $s\in\{F,C\}$ component. Recall $(U_n^{i,F}(n),U_n^{i,C}(n))$ for $n\ge 1$ is sampled from the kernel $\check{M}_n((\bar{u}_{n-1}^{i,F}(n-1),\bar{u}_{n-1}^{i,C}(n-1)),\uarg)$  where the $\bar{u}$ denotes post-resampling and the component $(U_n^{i,F}(j),U_n^{i,C}(j))=(\bar{u}_{n-1}^{i,F}(j),\bar{u}_{n-1}^{i,C}(j))$ for $j\in\{0,\dots,n-1\}$ is kept the same for the earlier components of the particle.

\begin{lemma}\label{lem:res3}
Assume (A\ref{hyp:1} (i) (iii), \ref{hyp:2}). Then for any $n\geq 0$ there exists a $M<\infty$ such that
$$
\mathbb{E}[\sum_{j=0}^n |U_n^{1,F}(j)-U_n^{1,C}(j)|^2] \leq M h_{\ell}^{\beta}.
$$
where $\beta$ is as in \eqref{eq:coup_h_cont}.
\end{lemma}

\begin{proof}
Our proof is by induction, the case $n=0$ following by \eqref{eq:coup_h_cont}. Assuming the result at $n-1$ we have
$$
\mathbb{E}[\sum_{j=0}^n |U_n^{1,F}(j)-U_n^{1,C}(j)|^2] = \mathbb{E}[\sum_{j=0}^{n-1} |\bar{U}_{n-1}^{1,F}(j)-\bar{U}_{n-1}^{1,C}(j)|^2 + |U_n^{1,F}(n)-U_n^{1,C}(n)|^2 ].
$$

Now
\begin{eqnarray*}
\mathbb{E}[\sum_{j=0}^{n-1} |\bar{U}_{n-1}^{1,F}(j)-\bar{U}_{n-1}^{1,C}(j)|^2] & = & N\sum_{j=0}^{n-1}\mathbb{E}\Big[
\frac{\check{G}_{n-1}(U_{n-1}^{1,F}(n-1),U_{n-1}^{1,C}(n-1))}
{\sum_{j=1}^N \check{G}_{n-1}(U_{n-1}^{j,F}(n-1),U_{n-1}^{j,C}(n-1))}\times \\ & & 
|U_{n-1}^{1,F}(j)-U_{n-1}^{1,C}(j)|^2\Big] \\ & \leq &
M\mathbb{E}[\sum_{j=0}^{n-1}|U_{n-1}^{1,F}(j)-U_{n-1}^{1,C}(j)|^2]
\end{eqnarray*}
where we have used (A\ref{hyp:1}) (i) and (iii). 
Applying the induction hypothesis along with \eqref{eq:coup_h_cont} yields
$$
\mathbb{E}[\sum_{j=0}^n |U_n^{1,F}(j)-U_n^{1,C}(j)|^2] \leq M\Big(h_{\ell}^{\beta} + \mathbb{E}[|\bar{U}_{n-1}^{1,F}(n-1)-\bar{U}_{n-1}^{1,C}(n-1)|^2]\Big)
$$
Now
$$
\mathbb{E}[|\bar{U}_{n-1}^{1,F}(n-1)-\bar{U}_{n-1}^{1,C}(n-1)|^2] = 
$$
$$
N \mathbb{E}\Big[\frac{\check{G}_{n-1}(U_{n-1}^{1,F}(n-1),U_{n-1}^{1,C}(n-1))}
{\sum_{j=1}^N \check{G}_{n-1}(U_{n-1}^{j,F}(n-1),U_{n-1}^{j,C}(n-1))}|U_{n-1}^{1,F}(n-1)-U_{n-1}^{1,C}(n-1)|^2\Big]
$$
Then by (A\ref{hyp:1}) (i) and (iii)
$$
\mathbb{E}\Big[\frac{\check{G}_{n-1}(U_{n-1}^{1,F}(n-1),U_{n-1}^{1,C}(n-1))}
{\sum_{j=1}^N \check{G}_{n-1}(U_{n-1}^{j,F}(n-1),U_{n-1}^{j,C}(n-1))}|U_{n-1}^{1,F}(n-1)-U_{n-1}^{1,C}(n-1)|^2\Big] \leq 
$$
$$
\frac{M}{N}\mathbb{E}[|U_{n-1}^{1,F}(n-1)-U_{n-1}^{1,C}(n-1)|^2] \leq
\frac{M}{N}\mathbb{E}[\sum_{j=0}^{n-1} |U_{n-1}^{1,F}(j)-U_{n-1}^{1,C}(j)|^2].
$$
Hence via the induction hypothesis, one has
$$
\mathbb{E}[|\bar{U}_{n-1}^{1,F}(n-1)-\bar{U}_{n-1}^{1,C}(n-1)|^2]  \leq M h_{\ell}^{\beta}
$$
and the proof is concluded.
\end{proof}

Recall Remark \ref{rem:key_point}.

\begin{proof}[Proof of Theorem \ref{theo:mainthm}]
This follows first by applying Proposition \ref{prop:gen_case}, followed by Lemma \ref{lem:res2} and then some standard calculations followed by Lemma \ref{lem:res3}.
\end{proof}

\begin{proof}[Proof of Corollary \ref{theo:maincor}]
Easily follows by adding and subtracting $\pmb{\check{\gamma}}_{n}(\psi)$ the $C_2$ inequality along with Theorem \ref{theo:mainthm}, and
then using \eqref{eq:main_id} combined with Proposition \ref{prop:h_smooth}.
\end{proof}

\section{Proof of consistency of the Markov chain Monte Carlo}
\label{app:method}
\begin{proof}[Proof of Theorem \ref{thm:consistency}] 
  Denote
  \begin{equation}\label{eq:xi}
    \xi_k(\funfullhelp) \defeq \big(\sum_{i=1}^\nrp V_k^{(i)} +
      \epsilon\big)^{-1} \big[\sum_{i=1}^{\nrp}
      V_{k}^{(i)}\funfullhelp(\Theta_k, X_k^{(i)}) +
      \sterm_{k}(\funfullhelp^{(\Theta_k)})\big],
\end{equation}
       where
      $\funfullhelp^{(\theta)}(x) \defeq \funfullhelp(\theta, x)$ and
       $\sterm_{k}(\funfullhelp^{(\theta)})\defeq p_{L_k}^{-1} \sum_{i=1}^{2 \nrp} V_{k,L_k}^{(i)} \funfullhelp^{(\theta)}( \mathbf{X}_{k,L_k}^{(i)}). 
       $
      Then
      $
      E_{\nri,\nrp,\mathbf{p}}(\funfullmain) = \frac{\sum_{k=1}^\nri \xi_k(\funfullmain)}{\sum_{j=1}^\nri \xi_k(\mathbf{1})}.
      $
Furthermore, by Assumption \ref{a:abstract-pf} \cite[cf.][]{rhee-glynn,vihola-unbiased},
we have
\begin{align*}
    \E[\sterm_{k}^2(\funfullhelp) \mid \Theta_k=\theta] & = s_\funfullhelp(\theta), \\
    \E[\sterm_{k}(\funfullhelp) \mid \Theta_k=\theta] &=
\gammass_\nrt^{(\theta,\infty)}(\gss_\nrt \funfullhelp)
-\gammass_\nrt^{(\theta,0)}(\gss_\nrt \funfullhelp)
\end{align*}
for $\funfullhelp = 1$ and $\funfullhelp = \funfullmain^{(\theta)}$.
This implies for $\funfullhelp = \funfullmain$ and $\funfullhelp = 1$, 
\begin{align*}
    \mu_\funfullhelp(\theta, v^{(1:\nrp)}, \,\mathbf{x}^{(1:\nrp)}) &\defeq 
    \E[\xi_k(\funfullhelp)\mid (\Theta_k, V_k^{(1:\nrp)}, \mathbf{X}_k^{(1:\nrp)})
    =(\theta,v^{(1:\nrp)},     \mathbf{x}^{(1:\nrp)})] \\
    &= \frac{1}{\sum_{j=1}^{\nrp} v^{(j)} + \epsilon}
    \bigg[\sum_{i=1}^{\nrp}
      v^{(i)} g(\theta, x^{(i)}) 
      - \gammass_\nrt^{(\theta,0)}(\gss_\nrt \funfullhelp)
      + \gammass_\nrt^{(\theta,\infty)}(\gss_\nrt \funfullhelp)\bigg], \\
    m_\funfullhelp^{(1)}(\theta, v^{(1:\nrp)}, \,\mathbf{x}^{(1:\nrp)})
    &\defeq \E[|\xi_k(g)|\mid (\Theta_k, V_k^{(1:\nrp)}, \mathbf{X}_k^{(1:\nrp)})
    =(\theta,v^{(1:\nrp)},     \mathbf{x}^{(1:\nrp)})]  \\
    &\le \frac{1}{\sum_{j=1}^{\nrp} v^{(j)} + \epsilon}
    \bigg[\sum_{i=1}^{\nrp}
      v^{(i)} |\funfullhelp(\theta, x^{(i)})|
      + \sqrt{\smash{s_{\funfullhelp^{(\theta)}}(\theta)}\vphantom{()}} \bigg].
\end{align*}
It is direct to check that the PMMH type chain $(\Theta_k,
X_k^{(1:\nrp)}, V_k^{(1:\nrp))})$ is reversible with respect to the
probability 
\begin{equation}    \label{eq:pmmh-invar}
    \pimcmc(\ud \theta, \ud x^{(1:\nrp)}, \ud v^{(1:\nrp)})
    = c_0 \mathrm{pr}(\theta) \ud \theta \mcmclatent_\theta^{(0)}(\ud x^{(1:\nrp)},
    \ud v^{(1:\nrp)}) \Big(\sum_{i=1}^{\nrp} v^{(i)} + \epsilon\Big),
\end{equation}
where $c_0>0$ is a normalisation constant and $\mcmclatent_\theta^{(0)}(\uarg)$
stands for the law of the output of Algorithm \ref{alg:pf} with
$(M_{0:\nrt}^{(\theta, 0)}$, $G_{0:\nrt}^{(\theta,0)},\nrp)$, and
therefore is Harris recurrent as a full-dimensional
Metropolis--Hastings that is $\psi$-irreducible \citep[cf.][Theorem
8]{roberts-rosenthal-harris}. It is direct to check that 
$\pimcmc(m_\funfullmain^{(1)})<\infty$, 
$\pimcmc(m_1^{(1)})<\infty$, $\pimcmc(\mu_\funfullmain) = c
\pi^{(\infty)}(\funfullmain)$ and $\pimcmc(\mu_1) = c$, where $c>0$
is a constant, so the result follows from \citep[Theorem
3]{vihola-helske-franks}.
\end{proof} 

\section{Proofs about asymptotic efficiency and allocations}
\label{app:efficiency}
\begin{proof}[Proof of Proposition \ref{prop:is-efficiency}]
  By Harris ergodicity,
  $
  \nrii^{-1} \cost(\nrii)
  \rightarrow
\mIRE
  $
almost surely. 
Dividing the inequality
  $$
\cost(\lmcmc(\kappa)) \le\kappa < \cost(\lmcmc(\kappa) +1)
$$
by $\lmcmc(\kappa)$ and taking the limit $\kappa\rightarrow \infty$, which implies $\lmcmc(\kappa)\rightarrow \infty$, we get that $\kappa/\lmcmc(\kappa)\rightarrow \mIRE$ almost surely.  Also, by Proposition \ref{prop:clt},
\[
    \sqrt{\lmcmc(\kappa)}\big[ E_{\lmcmc(\kappa),\nrp,\mathbf{p}}(\funfullmain) - \pi^{(\infty)}(\funfullmain) \big]
    \xrightarrow{\kappa\to\infty} \mathcal{N}(0,\sigma^2),\qquad \text{in distribution},
\]
so the result follows by Slutsky's theorem.    
\end{proof}

\begin{proof}[Proof of Proposition \ref{prop:efficiency-both-finite}]  We have that
  $$
  \E[\cost(\nrii)]
  =
  \sum_{k=1}^\nrii \E[\tau_{\Theta_k,L_k}]
  =
      \sum_{k=1}^\nrii \sum_{\ell=1}^\infty \E[\tau_{\Theta_k,\ell}] p_\ell.
      \le
C \nriii \sum_{k=1}^\infty p_\ell 2^{\costrate \ell(1+\rate)},      
      $$
by Assumption \ref{a:rates}\eqref{a:rates-cost}, 
which is finite if $r>\costrate(1+\rate)$.  Also,
$$
s_\funfullhelp(\theta)
=
\E[\sterm_k^2(\funfullhelp)| \Theta_k=\theta]
=
\sum_{\ell \ge 1} \frac{\E\Delta_\ell^2}{p_\ell}
\le
C \sum\Big(
2^{-\ell(\beta + \rate - r)} + 2^{-\ell(2\alpha -r)}
\Big),
$$      
which is finite if $r< \min(\beta+\rate, 2\alpha)$.  Therefore, $\sigma^2<\infty$, and the CLT follows by Proposition \ref{prop:is-efficiency}.
  \end{proof}

\begin{lemma}\label{lem:feller}
Let $\{X_k\}_{k\ge 1}$ be a sequence of independent random variables with $\E[X_{k_0}]=\infty$ for at least one $k_0$, and let $\{a_k\}_{k\ge 1}$ be a sequence of monotonically increasing real numbers with $a_k/k \longrightarrow \infty$.  Suppose one of the following assumptions holds:
\begin{enumerate}[(i)]
\item\label{lem:feller-iid}
  $
  \sum_{k\ge 1}  \P[ X_k > a_k]<\infty,
$
and $\{X_k\}_{k\ge 1}$ are also identically distributed, or
\item\label{lem:feller-indep}
  $
\sum_{k\ge 1}
\sup_{\nrii\ge 1} \P[ X_\nrii > a_k]
<
\infty.
$
\end{enumerate}
Then
$$
\P[ \sum_{k=1}^\nrii X_k > a_\nrii \;\text{infinitely many $\nrii\in\mathbb{N}$}]=0.
$$
\end{lemma}
\begin{proof}
\eqref{lem:feller-iid} is \citep[Theorem 2]{feller-infinite} since $\E[X_{k_0}]=\infty$ implies $\E[X_k]=\infty$ for all $k\ge 1$ as $\{X_k\}_{k\ge 1}$ are i.i.d.  For \eqref{lem:feller-indep}, note that if $X_k$ has c.d.f.~denoted $F_k$, then it is straightforward to check that
  $$
  F^*(x)
  \defeq
  \inf_{k\ge 1} F_k(x)
$$
  is a c.d.f.~also.  With $X^*_k\sim F^*$ i.i.d.~for $k\ge 1$, we have
  \begin{equation*}
  \P[X_k^* > a_k]
  =
  1- F^*(a_k)
  =
  \sup_{\nrii\ge 1} 1-F_\nrii(a_k)
  =
  \sup_{\nrii\ge 1} P[X_\nrii > a_k].
  \end{equation*}
  Summing over $k\ge 1$, we obtain $ \sum_{k\ge 1} \P[X_k^* > a_k] < \infty$.
  In addition,
  $$
  \E[X_k^*]
  =
  \int \P[X_k^* > x] \ud x
  \ge
  \int \P[X_{k_0} > x] \ud x
  =
  \infty,
  $$
for all $k\ge 1$.  Hence, we can apply \eqref{lem:feller-iid} for i.i.d.~random variables, obtaining
$$
0
=
\P[\sum_{k=1}^\nrii X^*_k > a_\nrii\;\text{infinitely many $\nrii$}]
\ge
\P[\sum_{k=1}^\nrii X_k > a_\nrii\;\text{infinitely many $\nrii$}],
$$
where the first equality comes from \eqref{lem:feller-iid}, and so we conclude.
\end{proof}

\begin{proof}[Proof of Proposition \ref{prop:subcanonical-abstract}]
  Conditional on output $\{\Theta_k\}_{k \ge 1}$ of Algorithm \ref{alg:is-mlmc}, $\{\tau_{\Theta_k,L_k} \}_{k\ge 1}$ are independent random variables.  Our assumptions imply Lemma \ref{lem:feller}\eqref{lem:feller-indep} holds, so
  \begin{equation*}
\P[\cost(\nrii) > a_\nrii \; \text{infinitely many $\nrii$}] =0, 
\end{equation*}
which means that $\cost(\nrii)$ is asymptotically bounded by $a_\nrii$.  Setting $\nrii=O(\epsilon^{-2})$ allows us to conclude.
  \end{proof}

The proofs below of Proposition \ref{prop:subcanonical-nonoptimised} and \ref{prop:subcanonical-optimised} are similar to that of \citep[Proposition 4 and 5]{rhee-glynn}.   
\begin{proof}[Proof of Proposition \ref{prop:subcanonical-nonoptimised}]
  With the prescribed choice of $p_\ell$ we have finite variance, as 
  $$
s_\funfullhelp(\theta)
  =
  \sum_{\ell \ge 1} \frac{\E \Delta_\ell^2 }{p_\ell}
  \le C
  \sum_{\ell \ge1} \frac{1}{\ell [ \log_2 (\ell +1)]^\eta}
  <\infty,
  $$
  uniformly in $\theta\in\mathsf{T}$.  To determine the order of complexity, we would like to apply Lemma \ref{lem:feller}\eqref{lem:feller-iid} to the i.i.d sequence $\{\tau_{L_k}^*\}_{k\ge 1}$, where $\tau_\ell^*\defeq C2^{\costrate \ell(1+\rate)}$. For any $k\ge 1$, where $a_k>0$ is some positive real number, we have,
  \begin{equation}\label{eq:feller-one}
  \P[\tau_{L_k}^* > a_k]
=
  \sum_{\ell\ge 1} \P[\tau_\ell^* > a_k] p_\ell
  =
  \sum_{\ell \ge 1} \mathbf{1}\left\{\ell > \frac{1}{\costrate(1+\rate)}\log_2\frac{a_k}{C}\right\} p_\ell.
  \end{equation}
  Because $\sum_{\ell \ge 1} p_\ell = 1$ and $p_\ell$ is monotonically decreasing, we have $\sum_{\ell \ge \ell_*} p_\ell$ is $O(p_{\ell_*})$.  Setting $\ell_*=\lfloor \frac{1}{\costrate(1+\rate)}\log_2 \frac{a_k}{C}\rfloor$, we therefore obtain that \eqref{eq:feller-one} is of order  
  $$
a_k^{-\frac{2b}{\costrate(1+\rate)} } \big(\log_2 a_k\big)\big( \log_2\log_2 a_k\big)^\eta.
  $$
Setting
\begin{equation}\label{eq:ak}
a_k\defeq [k (\log_2 k)^q]^\frac{\costrate(1+\rate)}{2 b}
\end{equation}
then ensures that $\sum_{k \ge 1} \P[\tau_{L_k}^* > a_k] <\infty$.  As $\beta\le 1$, it is easy to check that $\E[\tau_{L_k}^*]=\infty$.  We then apply Lemma \ref{lem:feller}\eqref{lem:feller-iid}, obtaining
$$
0
=
\P[ \sum_{k=1}^\nrii \tau_{L_k}^* > a_\nrii\;\text{infinitely many $\nrii$}]
\ge
\P[ \sum_{k=1}^\nrii \tau_{\Theta_k, L_k} > a_\nrii\;\text{infinitely many $\nrii$}].
$$
and conclude as before, by using that $\cost(\nrii)$ is asymptotically bounded by $a_\nrii$ and setting $\nrii=O(\epsilon^{-2})$.  
  \end{proof}

\begin{proof}[Proof of Proposition \ref{prop:subcanonical-optimised}]
We are in the basic setting of Proposition \ref{prop:subcanonical-nonoptimised} as before, but additionally may choose $\rate\ge 0$ as we please.  
The growth of $a_k$ given in \eqref{eq:ak} is essentially determined by $\costrate(1+\rate)/2b$, which can be made small when $\rate = 2\alpha-\beta$, implying $b=\alpha$.
  \end{proof}

\bibliographystyle{abbrvnat}


\end{document}